\documentclass[letter,11pt]{article}

\usepackage[english]{babel}
\usepackage[utf8x]{inputenc}
\usepackage[T1]{fontenc}
\usepackage{amsfonts}

\usepackage[margin = 1in]{geometry}

\usepackage{amsmath}
\usepackage{amsthm}
\usepackage{algpseudocode,algorithm,algorithmicx}
\usepackage{bbm}
\usepackage{graphicx}
\usepackage[colorinlistoftodos]{todonotes}
\usepackage[colorlinks=true, allcolors=blue]{hyperref}
\usepackage{thmtools}
\usepackage{thm-restate}
\usepackage{tikz}
\usetikzlibrary{arrows,shapes,snakes,automata,backgrounds,petri}

\newtheorem{lemma}{Lemma}
\newtheorem{observation}{Observation}
\newtheorem{definition}{Definition}
\newtheorem{corollary}[lemma]{Corollary}

\newcommand{\prob}{\textnormal{Pr}}
\newcommand{\Dt}{\widetilde{C}}
\newcommand{\Dtp}{\widetilde{C}'}

\newcommand{\depth}{\mathcal{D}}
\newcommand{\mdepth}{\mathcal{D}_{max}}
\newcommand{\st}{\widetilde{s}}
\newcommand{\pS}{p_{sub}}
\newcommand{\pI}{p_{ins}}
\newcommand{\pD}{p_{del}}
\newcommand{\pID}{p_{indel}}

\newcommand{\cond}{\mathcal{E}}
\newcommand{\ex}{\mathbb{E}}
\newcommand{\Var}{\mathrm{Var}}
\newcommand{\Cov}{\mathrm{Cov}}
\newcommand{\poly}{\text{poly}}
\newcommand{\plog}{\text{poly}\log(n)}
\newcommand{\set}[1]{\left\{ {#1} \right\}}

\newcommand{\eat}[1]{}
\newcommand{\replace}[2]{#2}
\global\long\def\code#1{\texttt{#1}}

\renewcommand{\widehat}{\hat}

\title{Optimal Sequence Length Requirements for Phylogenetic Tree Reconstruction with Indels}\author{
  Arun Ganesh
  \footnote{Department of Electrical Engineering and Computer Sciences, UC Berkeley. Email: \texttt{arunganesh@berkeley.edu}. Supported by NSF Award CCF-1535989.}
  \and
  Qiuyi (Richard) Zhang
  \footnote{Department of Mathematics, UC Berkeley. Email: \texttt{10zhangqiuyi@berkeley.edu}. Supported by NSF Award CCF-1535989.}
  }
\date{February 19, 2019}
\begin{document}
\maketitle

\begin{abstract}
We consider the phylogenetic tree reconstruction problem with insertions and deletions (indels). Phylogenetic algorithms proceed under a model where sequences evolve down the model tree, and given sequences at the leaves, the problem is to reconstruct the model tree with high probability. Traditionally, sequences mutate by substitution-only processes, although some recent work considers evolutionary processes with insertions and deletions. In this paper, we improve on previous work by giving a reconstruction algorithm that simultaneously has $O(\text{poly} \log n)$ sequence length and tolerates constant indel probabilities on each edge. Our recursively-reconstructed distance-based technique provably outputs the model tree when the model tree has $O(\text{poly} \log n)$ diameter and discretized branch lengths, allowing for the probability of insertion and deletion to be non-uniform and asymmetric on each edge. Our polylogarithmic sequence length bounds improve significantly over previous polynomial sequence length bounds and match sequence length bounds in the substitution-only models of phylogenetic evolution, thereby challenging the idea that many global misalignments caused by insertions and deletions when $p_{indel}$ is large are a fundamental obstruction to reconstruction with short sequences. 

We build upon a signature scheme for sequences, introduced by Daskalakis and Roch, that is robust to insertions and deletions. Our main contribution is to show that an averaging procedure gives an accurate reconstruction of signatures for ancestors, even while the explicit ancestral sequences cannot be reconstructed due to misalignments. Because these signatures are not as sensitive to indels, we can bound the noise that arise from indel-induced shifts and provide a novel analysis that provably reconstructs the model tree with $O(\text{poly} \log n)$ sequence length as long as the rate of mutation is less than the well known Kesten-Stigum threshold. The upper bound on the rate of mutation is optimal as beyond this threshold, an information-theoretic lower bound of $\Omega(\text{poly}(n))$ sequence length requirement exists.
\end{abstract}

\thispagestyle{empty}
\setcounter{page}{0}
\clearpage

\section{Introduction}

The phylogenetic tree reconstruction problem is a fundamental problem in the intersection of biology and computer science. Given a sample of DNA sequence data, we attempt to infer the phylogenetic tree that produced such samples, thereby learning the structure of the hidden evolutionary process that underlies DNA mutation. The inference of phylogenies from molecular sequence data is generally approached as a statistical estimation problem, in which a model tree, equipped with a model of sequence evolution, is assumed to have generated the observed data, and the properties of the statistical model are then used to infer the tree. Various approaches can be applied for this estimation, including maximum likelihood, Bayesian techniques, and distance-based methods \cite{warnow-textbook}.

Many stochastic evolution models start with a random sequence at the root of the tree and each child inherits a mutated version of the parent sequence, where the mutations occur i.i.d. in each site of sequence. The most basic model is the Cavender-Farris-Neyman (CFN) \cite{neyman1971molecular,cavender1978taxonomy} symmetric two-state model, where each sequence is a bitstring of 0/1 and mutations are random i.i.d. substitutions with probability $\pS(e)$ for each edge $e$. More complicated molecular sequence evolution models (with four states for DNA, 20 states for amino acids, and 64 states for codon sequences) exist but typically, the theory that can be established under the CFN model can also be established for the more complex molecular sequence evolution models used in phylogeny estimation \cite{essw2}. For simplicity, we will work with sequences that are bitstrings of 0/1 only.

In addition to computational efficiency, a reconstruction algorithm should also have a small {\it sequence length requirement}, the minimum sequence length required to provably reconstruct the model tree with high probability. Many methods are known to be statistically {\it consistent}, meaning that they will provably converge to the true tree as the sequence lengths increase to infinity, including maximum likelihood \cite{RochSly2017} and many distance-based methods \cite{atteson1999performance,warnow-textbook}. Methods that require only $O(\poly (n))$ sequence length are known as {\it fast converging} and recently, most methods have been shown to be fast converging under the CFN model, including maximum likelihood \cite{RochSly2017} (if solved exactly) and various distance-based methods \cite{essw1,essw2,warnow2001absolute,nakhleh2001designing,roch-science,mihaescu2013fast,brown2012fast}. More impressively, some algorithms achieve $O(\plog)$ sequence length requirement but require more assumptions on the model tree and they always need a tighter upper bound on $g$, the maximum edge length (in the CFN model, the length of an edge is defined as $\lambda(e) = -\ln(1-2\pS(e))$) \cite{mossel2004phase,roch2008sequence,roch-science,daskalakis2011evolutionary,brown2012fast,mihaescu2013fast}. Specifically, these methods are based on reconstruction of ancestral sequences and can provably reconstruct when $g$ is smaller than what is known as the Kesten-Stigum threshold, which is $\ln(\sqrt{2})$ \cite{mossel2004phase} for the CFN model. Intuitively, when $g$ is small, the edges have a small enough rate of mutation that allows for a concentration effect on estimators of ancestral sequences; however, when $g$ is past a certain threshold, a phase transition occurs and reconstruction becomes significantly harder. Essentially matching information-theoretic lower bounds show that $\Omega(\log (n))$ sequence lengths are needed to reconstruct when $g$ is smaller than the Kesten-Stigum threshold and $\Omega(\poly(n))$ lengths are needed beyond this threshold \cite{RochSly2017}.   

One of the biggest drawbacks of the CFN model is that it assumes that mutations only occur as a substitution and that bitstrings remain aligned throughout the evolutionary process down the tree. This allows for an relatively easy statistical estimation problem since each site can be treated as an i.i.d. evolution of a bit down the model tree. In practice, global misalignments from insertions and deletions (indels) breaks this site-independent assumption. Therefore, most phylogenetic tree reconstruction algorithms must first apply a multiple sequence alignment algorithm before attempting to reconstruct the sequence under a purely substitution-induced CFN model, with examples being CLUSTAL\cite{higgins1988clustal}, MAFFT\cite{katoh2002mafft}, and MUSCLE\cite{edgar2004muscle}. However, the alignment process is based on heuristics and lacks a provable guarantee. Even with a well-constructed pairwise similarity function, the alignment problem is known to be NP-hard \cite{elias2006settling}. Furthermore, it has been argued that such procedures create systematic biases \cite{loytynoja2008phylogeny,wong2008alignment}. 

Incorporating indels directly into the evolutionary model and reconstruction algorithm is the natural next step but the lack of site-wise independence presents a major difficulty. However, in a breakthrough result by \cite{daskalakis2010alignment}, the authors show that indels can be handled with $O(\poly(n))$-length sequences, using an alignment-free distance-based method. Also, in \cite{andoni2010phylogenetic}, the authors provide an $O(\plog)$-length sequence length requirement for tree reconstruction but can only handle indel probabilities of $\pID = O(1/\log^2(n))$, which is quite small since a string of length $O(\log(n))$ will only experience $O(1)$ indels as it moves $O(\log(n))$ levels down the tree. Similarly, \cite{andoni2012global} provides a method for reconstruction given $k$-length sequences as long as $\pID = O(k^{-2/3}(\log n)^{-1})$ for $k$ sufficiently large.  

In this paper, we almost close the gap by showing that provable reconstruction can be done in polynomial time with $O(\plog)$ sequence length rather than $O(\poly(n))$ sequence length, even when the probability of insertion and deletion are non-uniform, asymmetric, and $\pI, \pD$ are bounded by a constant. We do this by first constructing signature estimators that exhibit 1) robustness to indel-induced noise in expectation and 2) low variance when our mutation rate is below the Kesten-Stigum threshold. Then, these reconstructed signatures are the key components of a distance estimator, inspired by estimators introduced in \cite{roch2008sequence}, that uses $O(\log(n))$ conditionally independent reconstructed signatures to derive a concentration result for accurately estimating the distance between any nodes $a, b$. Our bounds on the rates of substitution, insertion and deletion are essentially optimal since we show that as long as our overall mutation rate is less than the Kesten-Stigum threshold, we can reconstruct with a sequence length requirement of $O(\plog)$, matching lower bounds up to a constant in the exponent. Our result implies that the noise introduced by insertions and deletions can be controlled in a similar fashion as the noise introduced by substitutions, breaking the standard intuition that misalignments caused by indels would inherently lead to a significantly higher sequence length requirement.

In Section~\ref{sec:prelim}, we provide a general overview of our model and methods. In Section~\ref{sec:sym}, we demonstrate that tree reconstruction with poly-logarithmic sequence length requirement for the symmetric case when the insertion and deletion probabilities are the same. In Section~\ref{sec:assym}, we extend our results to the asymmetric case and then we conclude with future directions.

\section{Preliminaries}
\label{sec:prelim}
\subsection{Model and Methods}

For simplicity, we'll consider the following model of phylogenetic evolution, which we call {\bf CFN-Indel}, as seen in \cite{andoni2010phylogenetic,daskalakis2010alignment}. We note that our analysis can be extended to the more general sequence models, such as GTR, with standard techniques. We start with a tree $T$ with $n$ leaves, also known as the model tree. There is a length $k$ bitstring at the root chosen uniformly at random from all length $k$ bitstrings. Each other node in $T$ inherits its parent's bitstring, except the following perturbations are made simultaneously and independently for each edge $e$ in the tree,

\begin{itemize}
\item Each bit is flipped with probability $\pS(e)$.
\item Each bit is deleted with probability $\pD(e)$.
\item Each bit inserts a random bit to its right with probability $\pI(e)$.
\end{itemize}

 \begin{definition}
 For an edge $e$, we denote that {\bf length} or {\bf rate of mutation} of that edge as 
 
 $$\lambda(e) = -[\ln(1 - 2\pS(e))+\ln(1-\pD(e))-\frac{1}{2}\ln(1+\pI(e)-\pD(e))]$$ 
 \end{definition}
 
Note that higher $\pI(e), \pD(e), \pS(e)$ leads to a higher edge length and that this $\lambda(e)$ is nonnegative: the only possibly negative term is the term $\frac{1}{2}\ln(1+\pI(e)-\pD(e))$, and if this term is negative, its absolute value is less than the term $-\ln(1-\pD(e))$. We provide some intuition for this definition: to estimate distances between nodes, our algorithm will look at the correlation between their bitstrings. We can show the correlation between two bitstrings decays by roughly $(1-2\pS(e))(1-\pD(e))$ for each edge $e$ on the path between the corresponding nodes, justifying the first two terms in the above definition. The term $\frac{1}{2}\ln(1+\pI(e)-\pD(e))$ is included to match a normalization term in our definition of correlation, which is needed to account for differing bitstring lengths throughout the tree. Thus, $\lambda(e)$ represents the rate of change that an edge produces in the sequence evolution process and therefore captures a notion of the length of an edge.  For two nodes, $a,b$, let $P_{a,b}$ denote the unique path between them and let $d(a,b) = \sum_{e \in P_{a,b}} \lambda(e)$ be our true distance measure between $a, b$. Note that $d$ is a tree metric or an {\bf additive} distance matrix. 
 
 As with all phylogenetic reconstruction guarantees, we assume upper and lower bounds $0 < \lambda_{min} \leq \lambda(e) < \lambda_{max}$ (in literature, $\lambda_{min}$ is often denoted with $f$ and $\lambda_{max}$ with $g$). Similar to \cite{roch2008sequence}, we work in the $\Delta$-branch model, where for all edges $e$, we have $\lambda(e) = \tau_e \lambda_{min}$ for some positive integer $\tau_e$. The {\it phylogenetic reconstruction problem} in the CFN-Indel model is to reconstruct the underlying model tree $T$ with high probability given the bitstrings at the leaves. 
 
 In this paper, we establish the claim that reconstruction can be done with $k = O(\log^\kappa n)$ bits given that $\lambda_{max}$ is the well known Kesten-Stigum threshold and $\lambda_{min} = \Omega(1/\plog)$. Note that this allows $\pS,\pD,\pI$ to be constant and that the upper bound $\lambda_{max}$ is information-theoretically optimal. We will first state the main theorem when $\pI(e) = \pD(e)$ for all edges $e$ and our model tree $T$ is balanced, which we call the {\it symmetric} case.

 \begin{restatable}{thm}{sym}
 \label{thm:symmetric}
 Assume that $\lambda_{max} \leq \ln(\sqrt{2})$ is less than the Kesten-Stigum threshold and $\lambda_{min} = \Omega(1/\plog)$. Furthermore, assume that we are in the symmetric case where $\pI(e) = \pD(e)$ for all edges $e$. Then for a sufficiently large $\kappa$, with $O(\log^\kappa n)$ sequence length, there is an algorithm \textsc{TreeReconstruct} that can reconstruct the phylogenetic tree with high probability under the CFN-Indel model.
 \end{restatable}
 
 Our high level idea is to estimate the additive distance matrix $d(a,b) = \sum_{e \in P_{a,b}} \lambda(e)$ and use standard distance-based methods in phylogeny to reconstruct the tree. As in \cite{daskalakis2010alignment}, our estimator of the distance relies on correlation calculations using blocks of consecutive sequence sites of length $l = \lfloor k^{1/2+\zeta}\rfloor$ for some small constant $0 < \zeta < 1/2$\footnote{The best setting of $\zeta$ will depend on other parameters introduced in the paper. For simplicity, the reader may wish to think of $\zeta$ as $1/4$.}. Therefore, we have approximately $L = \lfloor \frac{k}{l} \rfloor$ total disjoint blocks. For the bitstring at node $a$ and $1 \leq i \leq L$ let $\Delta_{a,i}$ be the signed difference between the number of zeroes in bits $(i-1)l + 1$ to $il$ of the bitstring at node $a$ (for convenience, we call this block $i$ of node $a$) and the expected number of zeroes, $\frac{l}{2}$.

\begin{definition} 
For a node $a$, we define the {\bf signature} of the corresponding sequence, $s_a$, as a vector with the $i$-th coordinate as
$$ s_{a,i} = \Delta_{a,i}/\sqrt{l}$$
\end{definition} 

We will use the signature of a sequence as the only information used in distance computations. Specifically, we note that signatures of two nodes far apart in the tree should have a low correlation, whereas the signatures of two nodes close together should have a high correlation. To realize this intuition, we prove concentration of signature correlations even under indel-induced noise and show that this concentration property can be applied recursively up the tree for signature reconstruction. 

Note that signatures are robust to indels as a single indel can only slightly change a signature vector, although it can have a global effect and change many signature coordinates at once. This leads to the somewhat accurate intuition that indels can introduce more noise than substitutions, as they can only produce local changes. However, because each coordinate of a signature, $s_{a,i}$, is an average over blocks of large size, we can control the indel-induced noise by 1) showing the signature is almost independent coordinate-wise and 2) applying concentration to produce an accurate distance estimator. Next, we introduce a novel analysis of indel-induced noise in signature reconstruction that allows for recursion up the tree. Specifically, we show that the indel-induced noise decays at about the same rate as the signal of the correlations between nodes, as we move down the tree. 

Putting it together, we are able to recursively reconstruct the signatures for all nodes using simply leaf signatures, showing that these estimators have low variance of $O(\log^2 n)$ as long as the edge length is less than the Kesten-Stigum threshold. Intuitively, this phenomenon occurs because the number of samples increases at a faster rate than the decay in correlation. By averaging over signatures in a sequence, this reduces the noise to $O(1/\plog)$. Finally, we show that this is sufficient for an highly accurate distance estimator via a Chernoff-type bound with $O(\log n)$ conditionally independent estimators to achieve tight concentration, leading to a recursive reconstruction algorithm via a simple distance-based reconstruction algorithm. 
 
In the asymmetric case, let $\mdepth$ be the maximum depth of our model tree. When $\pI(e) \neq \pD(e)$, we see that if $\pD(e) > \pI(e) + \kappa \frac{\log \log n}{\mdepth}$, our model will generate a nearly zero-length bitstring with high probability, as noticed in \cite{andoni2010phylogenetic}. Therefore, reconstruction is impossible. Furthermore, note that if $\mdepth > k^2$, where $k$ is the sequence length at the root, then the standard deviation in leaf sequence lengths due to insertion and deletion is on the order of $\Theta(\sqrt{\mdepth}) = \Omega(k)$ even if $\pD(e) = \pI(e)$. Again, we can easily encounter nearly zero-length bitstrings with decent probability. 

Otherwise, for any constants $\alpha, \beta > 0$ if we have $\mdepth \leq \log^{\alpha} n$ and $|\pD(e) - \pI(e)| \leq \frac{\beta \log \log n}{\mdepth}$, we show that if $\kappa$ is large enough\footnote{We note that the dependence of our $\kappa$ value on $\alpha, \beta$ does not match the previously mentioned lower bounds.} our algorithm can still reconstruct the underlying model tree, albeit through a more complicated analysis.
 
 \begin{restatable}{thm}{assym}
 \label{thm:asymmetric}
  Assume that $\lambda_{max}\leq \ln(\sqrt{2})$ and $\lambda_{min} = \Omega(1/\plog)$. Also assume for some constants $\alpha, \beta$ that $|\pI(e) - \pD(e)| \leq \frac{\beta \log \log n}{\mdepth}$ for all edges $e$, where $\mdepth \leq \log^{\alpha} n$ is the maximum depth of the model tree. Then for a sufficiently large constant $\kappa$, when the root has $O(\log^\kappa n)$ sequence length, there is an algorithm \textsc{TreeReconstruct} that can reconstruct the phylogenetic tree with high probability under the CFN-Indel model.
 \end{restatable}

\subsection{Notation} 

We give an overview of the notation we use most frequently here:

\begin{itemize}
\item $n$ is used to refer to the number of leaves in the model tree, $T$.
\item $\pS, \pD, \pI$ denote the probabilities of substitution, deletion, and insertion on an edge.
\item $\lambda_{min} \leq \lambda(e) \leq \lambda_{max}$ is the length or decay rate of an edge.
\item $k$ will be used to refer to the (polylogarithmic) length of a bitstring.
    \begin{itemize}
    \item $\kappa$ will be used to denote the exponent in the size of $k$.
    \end{itemize}
\item $a, b$ will be used to refer to nodes in the tree.
    \begin{itemize}
    \item $x$ will be used to refer to a leaf node.
    \item $\depth(a)$ will refer to the number of edges on the path from $a$ to the root, i.e. the depth of $a$.
    \item $a \wedge b$ will refer to the least common ancestor of $a, b$ in the tree.
    \item $d(a,b)$ will refer to the distance between two nodes in the tree.
    \item $\Dt(a, b)$ will refer to the ``correlation'' between the bitstrings at $a$ and $b$, and will be used to estimate $d(a,b)$.
    \item $A$ will be used to refer to the set of nodes which are descendants of $a$.
    \end{itemize}
\item $\sigma_{a, j}$ will be used to refer to the $j$th bit of the bitstring at $a$.
\item We will often split our bitstrings into ``blocks'' of bits, which are consecutive subsequences the bitstring. 
    \begin{itemize}
    \item We will use $i$ primarily to refer to the index of a block. 
        \begin{itemize}
        \item When we are referring to the index of something besides a block, we will use $j$ instead of $i$.
        \end{itemize}
    \item $l$ will be used to denote the length of these blocks, and $L$ the number of blocks.
    \item $s_{a,i}$ is a ``signature'' that will refer to the normalized (signed) difference between the number of zeroes in block $i$ of the bitstring at node $a$ and half its length.
        \begin{itemize}
        \item $\widehat{s}_{a,i}$ denotes an estimator of $s_{a,i}$ for internal nodes $a$.
        \item $\st_{a,i}$ denotes an estimate of $s_{a,i}$ for leaves in the asymmetric case.
        \end{itemize}
    \end{itemize}
\end{itemize}

Throughout the paper, we will use the following observation:

\begin{observation}
Let $X$ be any random variable, and $\cond$ an event which occurs with probability $1 - n^{-\Omega(\log n)}/B$, where $B$ is any upper bound on $|X|$. Then $\ex[X|\cond]$ and $\ex[X]$  differ by at most $n^{-\Omega(\log n)}$.
\end{observation}
\begin{proof}
$$\ex[X] = \ex[X|\cond]\Pr[\cond] + \ex[X|\lnot \cond]\Pr[\lnot \cond] = \ex[X|\cond](1-\frac{n^{-\Omega(\log n)}}{B})+\ex[X|\lnot \cond] \frac{n^{-\Omega(\log n)}}{B}$$

The observation follows because $\ex[X|\cond]\frac{n^{-\Omega(\log n)}}{B}$ and $\ex[X|\lnot \cond] \frac{n^{-\Omega(\log n)}}{B}$ are both at most $n^{-\Omega(\log n)}$ in absolute value.
\end{proof}

In this paper, all random variables we use can be upper bounded in magnitude by $O(\text{poly}(n))$. Thus for events $\cond$ which occur with probability $1-n^{-\Omega(\log n)}$, we may use $\ex[X|\cond]$ and $\ex[X]$ interchangeably as they only differ by at most $n^{-\Omega(\log n)}$, which will not affect any of our calculations. However, interchanges will often still be justified in proofs. 

\section{Reconstruction with Balanced Trees and Symmetric Probabilities}
\label{sec:sym}
In this section, we are in the symmetric case and assume $\pD(e) = \pI(e)$ for every edge. We also assume that the model tree is perfectly balanced. Both these assumptions will be relaxed later and our results are extended to the asymmetric case in the next section. 

We first demonstrate that some regularity conditions on the underlying bitstrings in the model tree hold with high probability. Given these regularity conditions, we show that the concentration for an recursive signature estimator provides a good distance estimator between any two nodes in the tree. Finally, we present our final distance-based reconstruction algorithm. Unless otherwise specified, proofs assume $\kappa$ is a sufficiently large constant depending only on $\zeta, \lambda_{min}$ and the parameters $\epsilon, \delta$ in the lemma/theorem statements. Our algorithmic construction will fix values of $\epsilon, \delta$, and thus works for some sufficiently large $\kappa$.

\subsection{High Probability Tree Properties}

Before we begin to describe our method for reconstructing the tree, we observe a few regularity properties about the bitstrings in the tree and prove that these properties hold with high probability. 

\begin{definition}

For an edge $(a, b)$ from parent node $a$ to child node $b$, we say that the $j_b$th bit of the bitstring at $b$ is \textbf{inherited} from the $j_a$th bit of the bitstring at $a$ if:

\begin{itemize}
    \item The $j_a$th bit of the bitstring at $a$ does not participate in a deletion on the edge $(a, b)$.
    \item The number of insertions minus the number of deletions on the edge $(a, b)$ in bits $1$ to $j_a - 1$ of the bitstring at $a$ is $j_b - j_a$.
\end{itemize}

For $a$ that is an ancestor of $b$, we extend this definition by saying that the $j_b$th bit of $b$ is inherited from the $j_a$th bit of the bitstring at $a$ if for the unique $a$-$b$ path $x_0 = a, x_1, \ldots x_k = b$, there are $j_0 = j_a, j_1 \ldots j_k = j_b$ such that for any $i$, the $j_i$th bit of the bitstring at $x_i$ is inherited from the $j_{i-1}$th bit of the bitstring at $x_{i-1}$.  

Lastly, to account for the case where $a = b$, we say that for any node $a$, the $j$th bit of $a$'s bitstring is \textbf{inherited} from the $j$th bit of $a$'s bitstring.
\end{definition}

\begin{definition}
For any two nodes $a, b$, the $j_a$th bit of the bitstring at $a$ and the $j_b$th bit of the bitstring at $b$ are \textbf{shared} if both are inherited from the $j$th bit of the bitstring at the least common ancestor of $a$ and $b$ for some $j$.
\end{definition}

\begin{definition}
For any two nodes $a$ and $b$, we say that the $j$th bit of the bitstring at $a$ \textbf{shifts} by $m$ bits on the path from $a$ to $b$ if there is $j'$ such that the $|j' - j| = m$ and the $j$th bit of the bitstring at $a$ and the $j'$th bit of the bitstring at $b$ are shared.
\end{definition}

It will simplify our analysis to assume all bitstrings are length at least $k$, which might not happen if the length of the root bitstring is $k$. Instead, we will let the length of the root bitstring be $2k$, and then by the following lemma all of the leaf bitstrings have $k$ bits to look at. 

\begin{lemma}\label{lemma:lengths}
If the bitstring at the root has length $2k$, then with probability $1 - n^{-\Omega(\log n)}$, the bitstring at all nodes have length at least $k$ and at most $4k$.
\end{lemma}

The next two regularity properties show that the bit shifts are in fact also small and that the number of excess zeros on a consecutive sequence of bits is small. They both follow from independence and concentration of binomials.

\begin{lemma}\label{lemma:bitshifts}
With probability $1 - n^{-\Omega(\log n)}$, no bit shifts by more than $4 \log^2 n \sqrt{k}$ bits on any path in the tree.  
\end{lemma}

\begin{lemma}\label{lemma:uniformblocks}
With probability $1-n^{-\Omega(\log n)}$, for all nodes $a$, the number of zeroes in any consecutive sequence of length $m$ sequence in $a$'s bitstring differs from $m/2$ by at most $\sqrt{m} \log n$. Consequently, $|s_{a,i}| \leq O(\log n)$.
\end{lemma}

We defer the proofs of Lemmas~\ref{lemma:lengths},~\ref{lemma:bitshifts}, and~\ref{lemma:uniformblocks} to Section~\ref{sec:deferred}.

\subsection{Distance Estimator}

We define $\cond_{reg}$ to be event that the high-probability regularity assumptions that are proven in Lemma~\ref{lemma:lengths},~\ref{lemma:bitshifts}, and~\ref{lemma:uniformblocks} all hold. Using correlations of signatures, we can define the distance estimator of two leafs $a,b$, $\Dt(a,b)$ analogously to \cite{daskalakis2010alignment}. 

$$\Dt(a, b) = \frac{2}{L}\sum_{i = 1}^{L/2} s_{a, 2i+1}s_{b, 2i+1}$$

In the case when indels do not occur, standard techniques can be used to show that for any leafs $a, b$, $\Dt(a,b)$ has an expectation that exponentially decays with respect to $d(a,b)$. The exponential decay comes from the observation that since the mutations can be viewed as a Markov transition from one state to the other, the correlations between states exponentially decays. Therefore, a back-of-the-envelope calculation gives $\ex[\Dt(a,b)] \approx \exp({-\sum_{e \in P(a,b)}\lambda(e)})\ex[\Dt(a,a)] \approx  \exp(-d(a,b))$.

We show that in the presence of indels, such an expectation still holds with $O(1/\plog)$ relative error. Key to our surprisingly small relative error, even when the insertion and probability errors are as large as a constant, is the observation that the indel-induced noise also decays exponentially with respect to $d(a,b)$.

\begin{lemma}\label{lemma:unbiasedestimator}
For any two nodes $a, b$ in the tree, and any $i$, $\ex[s_{a,i}s_{b,i}] = \frac{1}{4}(1\pm O(\log^{-\kappa\zeta+2} n))\exp(-d(a,b))$
\end{lemma}

We defer the proof to Section~\ref{sec:deferred}.

\begin{corollary}\label{lemma:unbiasedestimator2}
For any two nodes $a, b$, $\ex[\Dt(a,b)] = \frac{1}{4}(1\pm O(\log^{-\kappa\zeta+2} n))\exp(-d(a,b))$
\end{corollary}

Despite the indels, we can show that the odd-index blocks are almost independent conditioned on $\cond_{reg}$, i.e. in our analysis, we introduce a shift-invariant blockwise-independent signature scheme that is provably similar to our actual signature scheme. Thus, with high probability, we can also derive a tight concentration of our distance estimator that only uses signature correlations. We defer the proof and details to Section~\ref{sec:deferred}.

\begin{lemma}\label{lemma:estimatorconcentration}
Let $\delta > 0$ be any constant and $\epsilon = \Omega(\log^{\max\{-\kappa(1/2-\zeta)+2 \delta +6, -\kappa\zeta/2+\delta+3\}}(n))$. Then for nodes $a, b$ such that $ d(a, b) < \delta \log \log(n)$, then $|-\ln(4\Dt(a, b)) -  d(a, b)| < \epsilon$ with probability at least $1 - n^{-\Omega(\log n)}$.
\end{lemma}

\subsection{Signature Reconstruction}

The lemmas proven so far show that the estimator $-\ln(4\Dt(a, b))$ suffices to reconstruct the tree up to height $\delta \log \log n$. To reconstruct past this height, we will come up with a recursive estimator of the signatures of internal nodes, and then using this estimator build a more robust distance estimator for nodes at height above $\delta \log \log n$.

Let $L_h$ be the nodes at height $h = 0,..., \log n $, where $L_0$ contains all the leaves.  For a node $a \in L_h$, let $A$ be the set of leaves that are descendants of $a$ and we expect $|A| = 2^{h}$. Define the following estimator of $s_{a,i}$

\[\widehat{s}_{a,i} = \frac{1}{|A|}\sum_{x\in A} e^{d(x,a)} s_{x,i} \]

In the next few lemmas, we demonstrate that this signature estimator exhibits 1) robustness to indel-induced noise in expectation and 2) low variance when $\lambda_{max}$ is below the Kesten-Stigum threshold. The ultimate purpose of signature reconstruction is to introduce a distance estimator that uses $O(\log(n))$ conditionally independent reconstructed signatures to derive a concentration result for the distance between any nodes $a, b$. The definition of $\widehat{s}_{a,i}$ comes from similar intuition to that of Lemma~\ref{lemma:unbiasedestimator}: we expect that each edge on the path from $a$ to some descendant $x$ adds multiplicative decay to the correlation between $a$'s and $x$'s bitstrings, so as the next lemma formalizes, it should be that $\ex[s_{x,i}] \approx s_{a,i}e^{-d(x,a)}$.

\begin{lemma}\label{lemma:recursiveestimator}
Let $\cond$ denote the bitstring at the node $a$ and $\prob(\cond_{reg} | \cond) > 1 - n^{-\Omega(\log n)}$. Then, for a leaf $x$ that is a descendant of $a$, $\ex[s_{x,i} | \mathcal{E}] = e^{-d(x,a)}(s_{a,i} + \nu_{a,i})$, where $|\nu_{a,i}| \leq 8 \log^2 n k^{1/4}/\sqrt{l} = O(\log^{-\kappa\zeta/2+2} n)$.
\end{lemma}

We defer the proof to Section~\ref{sec:deferred}.

\begin{corollary}\label{cor:recursiveestimator}
Let $\mathcal{E}$ denote the bitstring at the node $a$ and $\prob(\cond_{reg} | \cond) > 1 - n^{-\Omega(\log n)}$. Then $\ex[\widehat{s}_{a,i} | \mathcal{E}] = s_{a,i} + \nu_{a,i}$, where $|\nu_{a,i}| \leq  O(\log^{-\kappa\zeta/2+2} n)$.
\end{corollary}

Next, we show that the signature estimator $\widehat{s}_{a,i}$ has $O(\log^2 n)$ variance, which relies on the fact that the variance reduction due to averaging is greater than the variance increase due to mutation when the mutation rate is less than the Kesten-Stigum threshold. It might be surprising that as we move up the tree, the variance of the estimator stays unchanged. However, since the correlations between two nodes exponentially decays in the distance, each term in the signature estimator becomes more ``independent'', allowing for a tight variance bound.

\begin{lemma}\label{lemma:signaturevariance}
Let $\mathcal{E}$ denote the bitstring at the node $a$ and $\prob(\cond_{reg} | \cond) > 1 - n^{-\Omega(\log n)}$. Then, $\ex[\widehat{s}_{a,i}^2|\cond] = O(\log^2 n)$ as long as $\lambda_{max} < \ln\sqrt{2}$.
\end{lemma}

\begin{proof}

$$ \ex[\widehat{s}_{a,i}^2| \mathcal{E}] = \frac{1}{|A|^2} \sum_{x,y \in A} e^{d(x,a) + d(y,a)} \ex[s_{x,i}s_{y,i}| \mathcal{E}] $$

To analyze $\ex[s_{x,i}s_{y,i}| \mathcal{E}]$, let $x \wedge y$ be the least common ancestor of $x, y$ and let $\mathcal{E'}$ denote the bitstring of $x \wedge y$. Note that conditioned on $\mathcal{E'}$, $s_{x,i}, s_{y,i}$ are independent and by Lemma~\ref{lemma:recursiveestimator}:

\begin{align*}
\ex[s_{x,i}s_{y,i} |\cond] &= \ex[\ex[s_{x,i}s_{y,i}| \mathcal{E}, \mathcal{E'}]] \\
&=\ex[\ex[s_{x,i}| \mathcal{E}']\ex[s_{y,i}| \mathcal{E}']] \\
&=e^{-d(x,y) }\ex[(s_{x\wedge y,i} + \delta_{x\wedge y,i})^2|\cond] \\
\end{align*}

Then, since $\cond_{reg}|\cond$ is a high probability event and noting that the quantity $(s_{x\wedge y,i} + \delta_{x\wedge y,i})^2$ is at most $\log^{2} n$ and thus conditioning on an event that happens with probability $1 - n^{-\Omega(\log n)}$ does not change its expectation by more than $n^{-\Omega(\log n)}$, we get $\ex[(s_{x\wedge y,i} + \delta_{x\wedge y,i})^2] \leq O(\log^2 n)$ and thus $\ex[s_{x,i}s_{y,i}| \mathcal{E}] \leq e^{-d(x,y)} \cdot O(\log^2 n)$.

This gives:
\begin{equation}
\begin{aligned}
\ex[\widehat{s}_{a,i}^2| \mathcal{E}] &\leq
\frac{1}{|A|^2} \sum_{x,y \in A} e^{d(x,a)+d(y,a)} \ex[s_{x,i}s_{y,i}| \mathcal{E}] \\&\leq
O(\log^2 n)\frac{1}{|A|^2} \sum_{x,y \in A} e^{2 d(a, x\wedge y)} \\
\end{aligned}
\end{equation}

Now, for a fixed $x$, note that $1/2$ of $y \in A$ satisfies $ d(a,x \wedge y) \leq \lambda_{max}$ and and $1/4$ of them satisfies $d(a,x\wedge y)\leq 2\lambda_{max}$ and so on. Since $e^{2\lambda_{max}} < 2$, 

\[\frac{1}{|A|} \sum_{y \in A} e^{2 d(a,x \wedge y)} \leq \left[(1/2)e^{2\lambda_{max}} + (1/4)e^{4\lambda_{max}} + ...\right] = O(1)\]

Finally, by symmetry, 

\[\ex[\widehat{s}_{a,i}^2| \mathcal{E}] \leq O(\log^2 n) \frac{1}{|A|} \sum_{x \in A} \frac{1}{|A|}\sum_{y \in A} e^{2d(a,x\wedge y)} = O(\log^2 n) \]
\end{proof}






\subsection{Distance Estimators}

Distance computations can be done with reconstructed signatures by analogously defining

\[\widehat{C}(a,b) = \frac{2}{L} \sum_{i=1}^{L/2} \widehat{s}_{a,2i+1}\widehat{s}_{b,2i+1}\]

Although we may use $\widehat{C}(a,b)$ directly as an estimator for the distance between $a, b$, the variance in the reconstructed signature is still too high for the necessary concentration. To provide the concentration, we use many conditionally independently distance estimators.

\begin{figure}
\begin{center}
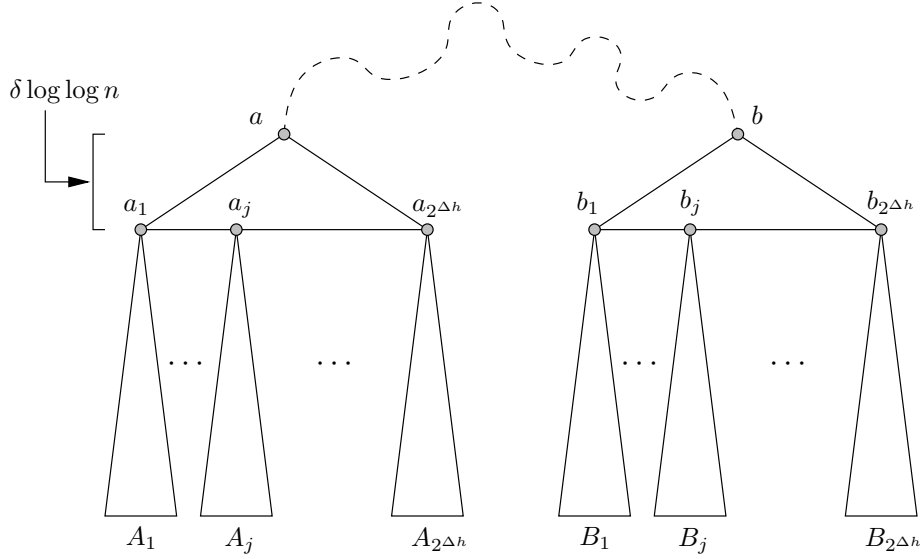\caption{Concentration via conditionally independent estimators. Slightly modified from a figure in \cite{roch2008sequence}}\label{fig:amplification}
\end{center}
\end{figure}

For two nodes $a, b$, consider creating the distance estimator $\widehat{ d}(a,b)$ as follows. For some height $\Delta h = \delta\log \log n$, consider the nodes that are descendants of $a, b$ exactly $\Delta h$ below $a, b$ respectively; order them arbitrarily as $a_1,...,a_{2^{\Delta h}}$ and $b_1,..., b_{2^{\Delta h}}$, as in Figure~\ref{fig:amplification}. Next, compute $\widetilde{d}(a_j,b_j) = -\ln(e^{d(a_j,a)+d(b_j,b)}4\widehat{C}(a_j, b_j))$, which is an estimator for $d(a,b)$. Note that we have $2^{\Delta h} = \Omega(\log n)$ of these estimators. We will aggregate these estimators in order to derive high probability concentration of the aggregate around $-\ln (4 \Dt(a,b))$, allowing us to use Lemma~\ref{lemma:estimatorconcentration} to show concentration of the aggregate around the true distance.

So, we proceed with the analogous construction that Roch presents for accuracy amplification for the substitution-only model \cite{roch2008sequence}. Consider the set of estimates $S_h(a,b) = \set{\widetilde{d}(a_j,b_j)}_{j=1}^{2^h}$. We will use a median-like measure to aggregate these estimates. For each $j$, we let $r_j$ be the minimum radius of an interval centered on  $\widetilde{d}(a_j,b_j)$ that captures at least $2/3$ of the other points in $S_h(a,b)$.

\[r_j = \inf\set{r > 0 : \left|\{j' \neq j: |\widetilde{d}(a_j, b_j) - \widetilde{d}(a_{j'}, b_{j'}) | \leq r  \}\right| \geq \frac{2}{3}2^{h}}\]

Then, if $j^* = \arg\min_j r_j$, then our distance estimator is $\widehat{ d}(a,b) = \widetilde{d}(a_{j^*},b_{j^*})$.

\begin{lemma}[Deep Distance Computation: Small Diameter]
\label{lemma:smalldiameter}
For any constant $\delta > 0$, let $a, b$ be nodes at height at least $\delta \log \log n$ such that $ d(a, b) \leq \delta \log \log n$. If $\lambda_{max} < \ln(\sqrt{2})$, $|\widehat{d}(a,b) - d(a,b)| < \epsilon$ with high probability.
\end{lemma}

\begin{proof}
Let $\xi$ denote the set of variables that are the underlying true bitstrings at nodes in $\{a_j\}_{j=1}^{2^{\Delta h}}$ and $\{b_j\}_{j=1}^{2^{\Delta h}}$. Throughout this proof, we will condition $\xi$ unless otherwise stated. Notice that we can translate all high probability results even upon conditioning. Note if the unconditioned probability $\prob(\cond_{reg}) > 1-n^{-\Omega(\log n)}$, the law of total expectation and a simple Markov bound shows us that $\prob(\prob(\cond_{reg} | \xi) > 1- n^{-\Omega(\log n)}) > 1- n^{-\Omega(\log n)}$, where the outer probability is taken over instantiation of $\xi$. This allows us to condition and establish independence of $\widehat{s}_{a_j,i}$ and $\widehat{s}_{b_j,i}$, while preserving high probability results.

By Lemma~\ref{lemma:recursiveestimator}, with high probability, $\ex[\widehat{s}_{a_j,i}] = s_{a_j,i} + \delta_{a_j,i}$ where  $|\delta_{a_j,i}| \leq O(\log^{-\kappa\zeta/2+2} n)$. Furthermore, by Lemma~\ref{lemma:uniformblocks}, $|s_{a_j,i}| \leq O(\log n)$ with high probability. Symmetrically, these bounds hold for $b_j$. Therefore, we see that 

\begin{equation}
\begin{aligned}
\ex[\widehat{C}(a_j,b_j)] &= \frac{2}{L} \sum_{i=1}^{L/2} \ex[\widehat{s}_{a_j,2i+1}\widehat{s}_{b_j,2i+1}] \\
&=  \frac{2}{L} \sum_{i}  s_{a_j,2i+1}s_{b_j,2i+1} + O(\log^{-\kappa\zeta/2+3} n) \\
&= \Dt(a_j,b_j) + O(\log^{-\kappa\zeta/2+3} n)
\end{aligned}
\end{equation}

\replace{

\begin{equation}
\begin{aligned}
\Var(\widehat{C}(a_j,b_j)) &\leq O(1/L^2) \sum_i  \Var(\widehat{s}_{a_j,2i+1}\widehat{s}_{b_j,2i+1}) \\
&\leq O(1/L^2) \sum_i  \ex[\widehat{s}_{a_j,2i+1}^2\widehat{s}_{b_j,2i+1}^2] \\
&\leq O(1/L^2) \sum_i  \ex[\widehat{s}_{a_j,2i+1}^2]\ex[\widehat{s}_{b_j,2i+1}^2] \\
&= O(\log^2 n/L) = O(\log^{-\kappa/2+\kappa\zeta+2}n)\\
\end{aligned}
\end{equation}

}{

Furthermore, we can bound the variance by using  Lemma~\ref{lemma:signaturevariance}. We first bound the covariance of two of the terms in the quantity $\widehat{C}(a_j, b_j)$:

\begin{lemma}\label{lemma:covariance}
Conditioned on $\xi$, for $i \neq i'$ $\Cov(\widehat{s}_{a_j,2i+1}\widehat{s}_{b_j,2i+1}, \widehat{s}_{a_j,2i'+1}\widehat{s}_{b_j,2i'+1}) = O(\log^{-\kappa\zeta/2+5} n)$
\end{lemma}

The proof is deferred to Section~\ref{sec:deferred}. Then the variance is bounded as follows:

\begin{equation}
\begin{aligned}
\Var(\widehat{C}(a_j,b_j)) &= O(1/L^2)\left[ \sum_i  \Var(\widehat{s}_{a_j,2i+1}\widehat{s}_{b_j,2i+1}) + \sum_{i \neq i'}\Cov(\widehat{s}_{a_j,2i+1}\widehat{s}_{b_j,2i+1}, \widehat{s}_{a_j,2i'+1}\widehat{s}_{b_j,2i'+1})\right]\\
&\leq O(1/L^2) \sum_i  \ex[\widehat{s}_{a_j,2i+1}^2\widehat{s}_{b_j,2i+1}^2]+O(\log^{-\kappa\zeta/2+5} n) \\
&\leq O(1/L^2) \sum_i  \ex[\widehat{s}_{a_j,2i+1}^2]\ex[\widehat{s}_{b_j,2i+1}^2] +O(\log^{-\kappa\zeta/2+5} n)\\
&= O(\log^2 n/L)+O(\log^{-\kappa\zeta/2+5} n) = O(\log^{-\kappa\zeta/2+5} n)\\
\end{aligned}
\end{equation}
}

Therefore, we can make the estimator variance $1/\plog$. Since $ d(a_j,b_j) \leq 2 \delta \log \log n + d(a,b) \leq  3\delta \log \log n$, by Lemma~\ref{lemma:estimatorconcentration}, we can guarantee w.h.p. that $(1+\epsilon)e^{- d(a_j,b_j)} \geq 4\Dt(a_j,b_j) \geq (1-\epsilon)e^{- d(a_j,b_j)}$ when $k$ is chosen with a large enough $\kappa$. Since $\epsilon$ is $\Omega(1/\plog)$, we see that $\ex[4\widehat{C}(a_j,b_j)] \in (1-2\epsilon,1+2 \epsilon) e^{-  d(a_j,b_j)}$ with constant probability by a Chebyshev bound for a fixed $j$. Therefore, we conclude that $|-\ln(4\widehat{C}(a_j,b_j)) - d(a_j,b_j)| < 2\epsilon$ with probability at least $5/6$. Since $d(a_j,b_j) = d(a,b) + d(a,a_j) + d(b,b_j)$, we have $|\widetilde{d}(a_j,b_j) - d(a,b)| < 2\epsilon$ with probability at least $5/6$.

Finally, since $a_j, b_j$ provide $\Omega(\log n)$ independent estimators of $d(a,b)$, by Azuma's inequality, we can show that at least $2/3$ of all $a_j, b_j$ satisfies $|\widetilde{d}(a_j,b_j) - d(a,b)| < 2\epsilon$ with probability at least $1-2^{-\Omega(\log n)} = 1- n^{-\Omega(1)}$. In particular, this means that there exists $j$ such that $r_j < 4 \epsilon$ and for all $j$ such that $|\widetilde{d}(a_j,b_j) -   d(a,b)| > 6  \epsilon$, we must have $r_j \geq 4  \epsilon$. Therefore, we conclude that $|\widetilde{d}(a_{j^*},b_{j^*}) - d(a,b)| = |\widehat{d}(a,b) - d(a,b)| < 6 \epsilon$.
\end{proof}

\subsection{Reconstruction Algorithm}

The algorithm for reconstruction is ultimately based from our ability to apply signature reconstruction and derive well-concentrated distance estimators in an inductive process. The base case would be to simply use the sequences at the leaves and the basic distance function to reconstruct the tree up to $O(\log \log n)$ height, after which we use our signature reconstruction algorithm to produce a reconstructed distance function that provides high accuracy throughout the entire process. 

In the previous section, we showed that if two nodes are $O(\log \log n)$ distance apart, then distance estimators will concentrate to the mean with $poly \log n$ sequence length. The recursive argument depends crucially that we can detect closeby nodes so that we only use statistically accurate distance estimators. Fortunately, testing for the size of the diameter of two nodes, on whether it is larger than or less than $O(\log \log n)$, is viable by the same concentration properties of our various distance estimators.

\begin{lemma}\label{lemma:testconcentration}
Let $\delta > 0$ be any constant and $\epsilon = \Omega(1/\plog)$. Then, for nodes $a, b$, if $ d(a, b) > r+ \epsilon$ with $r = \delta \log \log n$ and $k > \log^\kappa(n)$, then $-\ln(4\Dt(a, b)) > r$ with probability at least $1 - n^{-\Omega(\log n)}$.
\end{lemma}

\begin{proof}
Follows analogously to Lemma~\ref{lemma:estimatorconcentration}.
\end{proof}

\begin{definition} For two nodes $a, b$ that are more than $h  = \delta \log \log n$ up the tree, we define $T(S_h(a,b),r) = 1$ if at least half of $S_h(a,b)$ is bounded by $r$ and $0$ otherwise. 

$$T(S_h(a,b),r) = \mathbbm{1}\set{  |[-r, r] \cap S_h(a,b)| \geq \frac{1}{2}2^h} $$
\end{definition}

\begin{lemma}\label{lemma:deeptest}[Deep Computation: Diameter Test]
Let $a, b$ be nodes and we choose $r = O(\log \log n)$. If $k > \plog $, then with high probability, $T(S_h(a,b), r) = 1$ when $d(a,b) < r -\epsilon$ and $T(S_h(a,b), r) = 0$ when $d(a,b) > r+\epsilon$.
\end{lemma}

\begin{proof}
The case when $d(a,b) < r - \epsilon$ follows directly from the proof of Lemma~\ref{lemma:smalldiameter}. When $d(a,b) > r + \epsilon$, note that the only change to the proof of Lemma~\ref{lemma:smalldiameter} is that instead of calling Lemma~\ref{lemma:estimatorconcentration} to upper and lower bound $\Dt(a_j,b_j)$, we use Lemma~\ref{lemma:testconcentration} to deduce that $\widetilde{d}(a_j,b_j) \geq r+\epsilon/2$ still holds with constant probability by Chebyshev and $T(S_h(a,b),r) = 0$ occurs with high probability using Azuma's inequality bound over all pairs $(a_j, b_j)$.
\end{proof}

With the diameter test, we can ensure that all distance computations are accurate with an additive error of $\epsilon$ by using a diameter condition and also guarantee that all close enough nodes have distances computed. Therefore, we can compute an accurate localized distance matrix. With that, we use the traditional Four Point Method to determine quartets. The standard technique, called the Four Point Method, to compute quartet trees (i.e., unrooted binary trees on four leaves) is based on the Four Point Condition \cite{buneman-4pt}. The underlying combinatorial algorithm we use here is essentially identical to the one used by Roch in \cite{roch2008sequence}.

 \begin{definition} (From \cite{essw1})
Given a four-taxon set $\{a, b, c, d\}$ and a dissimilarity matrix $D$, the Four Point Method ($\code{FPM}$) infers tree  $ab|cd$ (meaning the quartet tree with an edge separating $a, b$ from $c, d$) if $D(a, b) + D(c, d) \leq \min \{ D(a,c) + D(b,d), D(a,d) + D(b,c)\}$. If equality holds, then the $\code{FPM}$ infers an arbitrary topology. 
\end{definition}

If $D(a,b)$ is a dissimilarity matrix that has maximum deviation from $d(a,b)$ by an additive error of $\epsilon < f/2$, where $f$ is the minimum non-zero entry in $d$, then $\code{FPM}$ will always infer the true quartet, as in Figure~\ref{fig:induced}. In this case, setting $\epsilon < \lambda_{min}/2$ will allow for correct short quartet inference and therefore this implies that $\lambda_{min} = \Omega(1/\plog)$ in order for the sequence length requirement to still be polylogarithmic (in general it will depend inverse polynomially on $\epsilon$). 

Once quartet splits are determined accurately, any quartet-based tree-building algorithm can be used. For simplicity, we will use a cherry picking algorithm that simply identifies $a,b$ as a cherry if they are always on the same side of all quartet splits. Then, we can reconstruct the ancestors of these cherries and simply recurse. This is the high-level summary of our reconstruction algorithm~\ref{alg:sym-reconstruct}. 

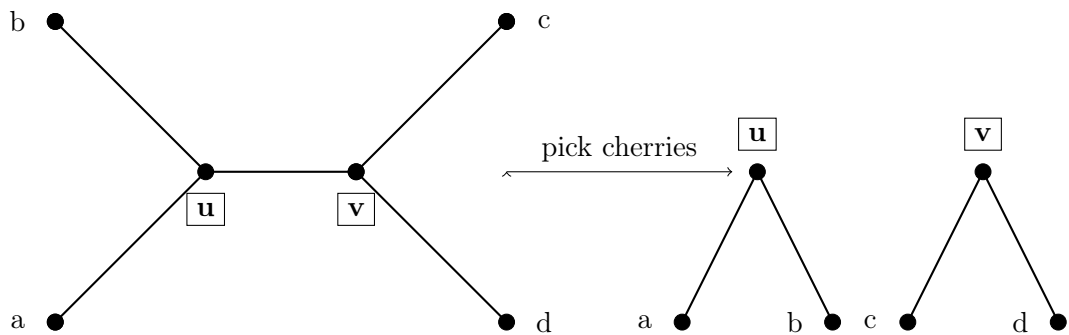
\begin{figure}
\begin{tikzpicture}
  \begin{scope}
    \draw[fill=black] (4,3) circle (3pt);
    \draw[fill=black] (2,1) circle (3pt);
    \draw[fill=black] (2,5) circle (3pt);
    \draw[fill=black] (6,3) circle (3pt);
    \draw[fill=black] (8,5) circle (3pt);
    \draw[fill=black] (8,1) circle (3pt);
    \node at (1.5,1) {a};
	\node at (1.5,5) {b};
    \node at (8.5,5) {c};
    \node at (8.5,1) {d};
    \node[draw, font=\bf] at (6,2.5) {v};
    \node[draw, font=\bf] at (4,2.5) {u};
    \draw[fill=black] (2,1) circle (3pt);
    \draw[fill=black] (2,5) circle (3pt);
    \draw[fill=black] (6,3) circle (3pt);
    \draw[fill=black] (8,5) circle (3pt);
	\draw [thick] (2,1) -- (4,3) -- (2,5);
    \draw [thick] (4,3) -- (6,3);
    \draw[thick] (8,1) -- (6,3) -- (8,5);
  \end{scope}
  
  \draw[->] (8,3) edge node[above]{pick cherries} (11,3);
  \begin{scope}[xshift=2.1in]
    \draw[fill=black] (6,3) circle (3pt);
    \draw[fill=black] (9,3) circle (3pt);
    \draw[fill=black] (5,1) circle (3pt);
    \draw[fill=black] (7,1) circle (3pt);
    \draw[fill=black] (8,1) circle (3pt);
    \draw[fill=black] (10,1) circle (3pt);
    \node[draw, font=\bf] at (6,3.5) {u};
     \node[draw, font=\bf] at (9,3.5) {v};
    \node at (4.5,1) {a};
	\node at (6.5,1) {b};
	\node at (7.5,1) {c};
	\node at (9.5,1) {d};
    \draw[thick] (5,1) -- (6,3) -- (7,1);
    \draw[thick] (8,1) -- (9,3) -- (10,1);
  \end{scope}
\end{tikzpicture}
\caption{The Four Point Method on $\set{a,b,c,d}$ with distance matrix $D$ will infer the correct quartet as long as $D(a,b)$ do not differ from the underlying additive distance $d(a,b)$ by more than $\lambda_{min}/2$. We then use quartet splits to pick cherries and recurse on ancestors (i.e. $a, b$)}
\label{fig:induced}
\end{figure}

\begin{definition}
A quartet  $Q = \set{a,b,c,d}$ of $L_h$ is {\bf r-short} if

\begin{itemize}
    \item When $h = 0$, we use $D(x,y) = -\ln(4\widetilde{C}(x,y))$ and $$\max_{x, y \in Q} D(x,y) \leq r$$
    \item When $h > \delta \log \log n$, we use $D(x,y) = \widehat{d}(x,y)$ and $$ \min_{x,y \in Q} T(S_h(x,y), r) = 1$$
\end{itemize} 
\end{definition}

\begin{lemma}
\label{lemma:shortquartet}
If a quartet $Q$ is $O(\log \log n)$-short and $\lambda_{min}=\Omega(1/\plog)$, then with high probability, if $k > \plog$, $\code{FPM}$ with the corresponding distance matrix $D$ will return the true quartet tree.
\end{lemma}

\begin{proof}
Since $\code{FPM}$ will return the true quartet tree when the distance matrix is $\epsilon < \lambda_{min}/2$ away from the true distance matrix, this follows directly from combining Lemma~\ref{lemma:testconcentration}, ~\ref{lemma:estimatorconcentration} for the case when $h = 0$. We use Lemma~\ref{lemma:deeptest},~\ref{lemma:smalldiameter} for the case when $h > \delta \log \log n$.
\end{proof}

When we recursively build this tree up, it is crucial that we can calculate distances between nodes in each sub-tree that is reconstructed so far. This is because our distance estimators $\widehat{d}(a,b)$ requires knowledge or a good estimate of all distances in the subtree under $a$ and $b$ in order to calculate the $O(\log n)$ conditionally independent low-variance distance estimates. Fortunately, it suffices to have a good enough estimate of these distances when distances can only take only discrete integer multiples of $\lambda_{min}$, as assumed in our $\Delta$-branch model. Under this assumption, we can ascertain the distances by simply rounding to the nearest integer multiple of $\lambda_{min}$ as long as $\epsilon < \lambda_{min}/2$. The estimation of all relevant distances is based of a very simple three-point rule.

\begin{definition}[Three-Point Rule]
For a triplet of nodes $a,b,c$ that meet at $x$ and a dissimilarity matrix $D$, we define $\widehat{D}(a,x)$ to be the estimator 

\[ \widehat{D}(a,x) = \frac{1}{2}[D(a,b) + D(a,c) - D(b,c)]\]
\end{definition}

We are now ready to present the final algorithm, $\textsc{TreeReconstruct}$. We assume that the input to the algorithm are just the leaf signatures $s = \set{s_{x}}$ and $\lambda_{min}$.

\begin{algorithm}
  \caption{Tree Reconstruction With Signatures
    \label{alg:sym-reconstruct}}
  \begin{algorithmic}[1]
    \Require{Leaf signatures $\set{s_{x}}$ for all $x \in L_0$, $\lambda_{min}$ }
    \Statex
    \Function{TreeReconstruct}{$\set{s_{x}}$, $\lambda_{min}$}
      \State Apply $\code{FPM}$ to $2\delta\log\log n$-short quartets of $L_0$ with $D(a, b) = -\ln(4\widetilde{C}(a, b))$ to infer splits
      \State Use quartet splits and the three-point rule to build multiple subtrees up to $\delta \log \log n$ height \par by iteratively identifying cherries.
      \State Estimate and fill in distances within each subtree using the three-point rule, rounded to the \par nearest $\lambda_{min}$ multiple.
      \For{$h  \gets  \delta \log \log n \textrm{ to } \log n$}
        \State Apply $\code{FPM}$ to $2\delta \log \log n$-short quartets on $L_{h}$ with $D(a, b) = \widehat{d}(a, b)$ to infer splits.    
        \State Identify cherries as pairs of vertices that only appear on the same side of quartet splits. 
        \State Estimate and fill in the edge lengths on cherries using the three-point rule on a short \par\hskip0.5cm quartet, rounded to nearest $\lambda_{min}$ multiple.
        \State Add cherries, with edges containing estimated decay rates, to the tree
      \EndFor
      \State \Return{the resulting tree $t$}
    \EndFunction
  \end{algorithmic}
\end{algorithm}

\sym*

\begin{proof}
Set $\epsilon = \lambda_{min}/3 = \Omega(1/\plog)$. Then, by Lemma~\ref{lemma:shortquartet}, all quartets queries made by \textsc{TreeReconstruct} are correct. Note that $a, b$ that are neighbors (i.e. for which $a, b$ is a cherry) appear on the same side of all short quartets. Furthermore, if $a, b \in L_h$ are not cherries and $d(a,b) < O(\log \log n)$, then there must exists $x, y$, such that $\set{a,x,b,y}$ is $O(\log \log n)$-short and $a,b$ are not on the same side of the split. Otherwise, if $d(a,b) = \Omega(\log \log n)$ is large, then by Lemma~\ref{lemma:testconcentration} and Lemma~\ref{lemma:deeptest}, none of the quartets involving $a,b$ will be considered $O(\log \log n)$-short. Finally, since all short quartets are considered, we conclude that all cherries picked are correct. 

Lastly, we note that we can estimate distances up to error $\epsilon < \lambda_{min}/2$, which allows us to round to the nearest multiple of $\lambda_{min}$ to get noiseless distance reconstruction with high probability. Therefore, all estimated decay rates or lengths are correct, as we reconstruct up the tree, by using Lemma~\ref{lemma:smalldiameter}, and noting that if all distances in $D$ are accurate up to error $\epsilon$, then the Three-Point Rule is accurate up to error at most $3\epsilon/2 < \lambda_{min}/2$.
\end{proof}

\section{Unbalanced Trees and Asymmetric Probabilities}
\label{sec:assym}
In this section, we show how to relax the assumption that the tree is completely balanced. Instead, we consider trees which are approximately balanced. In particular, let $\depth(a)$ denote the number of edges on the path from $a$ to the root. Then for $\mdepth$ such that $\mdepth \leq \log^{\alpha} n$ for a constant $\alpha$, we assume all nodes satisfy $\depth(a) \leq \mdepth$.

We also relax the assumption that $\pD(e) = \pI(e)$ for every edge, instead assuming that $|\pI(e) - \pD(e)|$ is bounded by $\beta \frac{\log \log n}{\mdepth}$ for a constant $\beta$. As mentioned in Section~\ref{sec:prelim}, up to the constants $\alpha, \beta$ these assumptions are optimal, i.e. for a fixed $\kappa$ and sufficiently large $\alpha$ or $\beta$, reconstruction with high probability is not possible due to significant loss of bitstring length down the tree. We will assume for simplicity of presentation the length of the root bitstring is $\Theta(\log^\kappa n)$, but the analysis easily generalizes to the case where the root has a larger bitstring.

Again, in all proofs we assume $\kappa$ is a sufficiently large constant depending only on the fixed values $\alpha, \beta, \zeta, \lambda_{min}$ and parameters $\delta, \epsilon$ in the lemma/theorem statements (which will be fixed by our algorithmic construction).

\subsection{Tree Properties}

Let $k_a$ be the number of bits in the bitstring at vertex $a$, and $k_r = \Theta(\log^\kappa n)$ specifically be the number of bits in the root bitstring. Let $L = \lfloor k_r^{1/2-\zeta} \rfloor$ for some small constant $\zeta > 0$. The length of a block at the root $l_r$ will be $\lfloor k_r/L \rfloor$ as before. 

Then, define $\eta(a) = \prod_{e \in P_{r, a}}(1 + \pI(e) - \pD(e))$. Note that the expected position of bit $j$ of $r$ in $a$ conditioned in the bit not being deleted is $j\eta(a)$. Thus, we will define the length of a block in bitstring $a$ to be $l_a = \lfloor l_r \eta(a) \rfloor$, and the $i$th block of the bitstring at node $a$ to be bits $(i-1)l_a+1$ to $il_a$ of the bitstring. Note that by the assumption that $|\pI(e) - \pD(e)| \leq \beta \frac{\log \log n}{\mdepth}$, $\log^{-\beta}(n) \leq \eta(a) \leq \log^\beta(n)$ for all $a$.

Key to our algorithmic construction is the fact that not only do we expect $k_a = k_r \eta(a)$, but that this concentrates very tightly.

\begin{lemma}\label{lemma:lengths-asym}
With probability $1 - n^{-\Omega(\log n)}$ for all vertices $a$, $k_r  (1 - \frac{\depth(a) \cdot \log^{\beta/2+1}(n)}{\log^{\kappa/2-2} n}) \leq k_a/\eta(a) \leq k_r (1 + \frac{\depth(a)\cdot \log^{\beta/2+1}(n)}{\log^{\kappa/2-2} n})$
\end{lemma}

We defer the proof to Section~\ref{sec:deferred-assym}. For the purposes of analysis, it will be convenient to define the \textit{normalized shift} of a bit. 
\begin{definition}
For any two nodes $a$ and $b$, we say that the $j$th bit of the bitstring at $a$ has a \textbf{normalized shift} of $m$ bits on the path from $a$ to $b$ if there is some $j'$ such that $|j'/\eta(b) - j/\eta(a)| = m$ and the $j$th bit of the bitstring at $a$ and the $j'$th bit of the bitstring at $b$ are shared.
\end{definition}

\begin{lemma}\label{lemma:bitshifts-asym}
With probability $1 - n^{-\Omega(\log n)}$, no bit has a normalized shift of more than $4 \log^{\alpha+1} n \sqrt{k_r}$ bits on any path in the tree.  
\end{lemma}

We defer the proof to Section~\ref{sec:deferred-assym}. Analogously to before, we will define $\cond_{reg}$ to be the intersection of the high probability events described in Lemmas~\ref{lemma:lengths-asym}, \ref{lemma:bitshifts-asym}, and \ref{lemma:uniformblocks}.

We define $s_{a,i}$ analogously to before, letting it be $1/\sqrt{l_a}$ times the signed difference between the number of zeroes in the $i$th block of the bitstring of the bitstring at node $a$ and half the length of $i$th block, and note that Lemma~\ref{lemma:uniformblocks} still applies. However, we do not know the true block lengths, so even for the leaves we cannot exactly compute $s_{a,i}$. 

Instead, for a leaf bitstring let $l_a' = \lfloor k_a/L \rfloor$. Note that this quantity is computable given only the leaf bitstrings as well as the minimum sequence length requirement $k_r$. Our algorithm will split each leaf bitstring into ``pseudo-blocks'' of length $l_a'$, i.e. the $i$th pseudo-block of leaf $a$ consists of bits $(i-1)l_a'+1$ to $il_a'$ of the bitstring. Our estimate $\st_{a,i}$ of $s_{a,i}$ is then $1/\sqrt{l_a'}$ times the signed difference between the number of zeroes in the $i$th pseudo-block of the bitstring at node $a$. $\st_{a,i}$ is computable for all leaf nodes $a$ since we can compute $l_a'$ easily, so giving a reconstruction algorithm based on the signatures $\st_{a,i}$ gives a constructive result.

\begin{lemma}\label{lemma:signatureestimation}
Conditioned on $\cond_{reg}$, $\st_{a,i} = s_{a,i} \pm O(\log^{\alpha/2+\beta/4 + 5/2 - \kappa\zeta/2} n)$
\end{lemma}

We defer the proof to Section~\ref{sec:deferred-assym}.

\subsection{Distance Estimator}

We would like to compute the following estimator as before: 

$$\Dt(a, b) = \frac{2}{L}\sum_{i = 1}^{L/2} s_{a, 2i+1}s_{b, 2i+1}$$

But we cannot directly compute $s_{a,i}$, so we use the following estimator instead:

$$\Dtp(a, b) = \frac{2}{L}\sum_{i = 1}^{L/2} \st_{a, 2i+1}\st_{b, 2i+1}$$

Note that by Lemma~\ref{lemma:signatureestimation} and Lemma~\ref{lemma:uniformblocks} we immediately get the following Corollary:

\begin{corollary}\label{cor:estimatorsimilarity}
Conditioned on $\cond_{reg}$, $|\Dtp(a, b) - \Dt(a,b)| = O(\log^{\alpha/2+\beta/4 + 7/2 - \kappa\zeta/2} n)$
\end{corollary}

\begin{lemma}\label{lemma:unbiasedestimator-asym}
For any two nodes $a, b$ in the tree, and any $i$, $$\ex[s_{a,i}s_{b,i}] = \frac{1}{4}(1\pm O(\log^{-\kappa\zeta+O(\log^{-\kappa\zeta+\alpha+1} n))} n))\exp(-d(a,b))$$
\end{lemma}

The proof follows similarly to the proof of Lemma~\ref{lemma:unbiasedestimator} and is deferred to Section~\ref{sec:deferred-assym}.

\begin{corollary}\label{lemma:unbiasedestimator2-asym}
For any two nodes $a, b$, $\ex[\Dt(a,b)] = \frac{1}{4}(1\pm O(\log^{-\kappa\zeta+\alpha+1} n))\exp(-d(a,b))$.
\end{corollary}

\begin{lemma}\label{lemma:estimatorconcentration-asym}
Let $\delta > 0$ be any constant and $\epsilon =  \Omega(\log^{max\{-\kappa(1/2-\zeta)+2 \delta +6, -\kappa\zeta/2+\alpha/2+\delta+5/2\}} n)$. Then for nodes $a, b$ such that $ d(a, b) < \delta \log \log(n)$, then $|-\ln(4\Dt(a, b)) -  d(a, b)| < \epsilon$ with probability at least $1 - n^{-\Omega(\log n)}$.
\end{lemma}

\begin{proof}
The proof follows exactly as did the proof of Lemma~\ref{lemma:estimatorconcentration}.

\end{proof}

\begin{corollary}\label{cor:estimatorconcentration-asym}
Let $\delta > 0$ be any constant and $\epsilon = \Omega(\log^{\max\{-\kappa(1/2-\zeta)+2 \delta +6,- \kappa\zeta/2+\alpha/2+ \beta/4 +\delta+ 7/2 \}} n)$. Then for nodes $a, b$ such that if $ d(a, b) < \delta \log \log(n)$, then $|-\ln(4\Dtp(a, b)) -  d(a, b)| < \epsilon$ with probability at least $1 - n^{-\Omega(\log n)}$.
\end{corollary}
\begin{proof}
The proof follows from Lemma~\ref{lemma:estimatorconcentration-asym} and Corollary~\ref{cor:estimatorsimilarity}.
\end{proof}

\subsection{Signature Reconstruction}

For a node $a$ in the tree, let $A$ be the set of leaves that are descendants of $a$. Let $h(a)$ be the maximum number of edges on the path between $a$ and any of its leaves, i.e. $h(a) = \max_{x' \in A} \depth(x') - \depth(a)$. For for a leaf $x \in A$ let $h(x) = \max_{x' \in A} \depth(x') - \depth(x)$, i.e. $h(x)$ is the difference between $x$'s depth and the maximum depth of any leaf. If $a$ is far away from all of its leaf descendants, we will use the following estimator of its signature:

\[\widehat{s}_{a,i} = \frac{1}{2^{h(a)}}\sum_{x\in A} e^{d(x,a)}2^{h(x)} \st_{x,i} \]

Note that $\sum_{x \in A} 2^{h(x)} = 2^{h(a)}$. In addition, if the tree below $a$ is balanced then $h(x) = 0$ for all $x \in A$ so the estimator is defined analogously to before.

\begin{lemma}\label{lemma:recursiveestimator-asym}
Let $\cond$ denote the bitstring at the node $a$ and $\prob(\cond_{reg} | \cond) > 1 - n^{-\Omega(\log n)}$. Then, for a leaf $x$ that is a descendant of $a$, $\ex[\st_{x,i} | \mathcal{E}] = e^{-d(x,a)}(s_{a,i} + \nu_{a,i})$, where $|\nu_{a,i}| = O(\log^{\alpha/2+\beta/4+5/2-\kappa\zeta/2} (n))$.
\end{lemma}

The proof follows similarly to the proof of Lemma~\ref{lemma:recursiveestimator} and is deferred to Section~\ref{sec:deferred-assym}.

\begin{corollary}\label{cor:recursiveestimator-asym}
Let $\mathcal{E}$ denote the bitstring at the node $a$ and $\prob(\cond_{reg} | \cond) > 1 - n^{-\Omega(\log n)}$. Then $\ex[\widehat{s}_{a,i} | \mathcal{E}] = s_{a,i} + \nu_{a,i}$, where $|\nu_{a,i}| \leq  O(\log^{\alpha/2+\beta/4+5/2-\kappa\zeta/2} (n))$.
\end{corollary}

\begin{lemma}\label{lemma:signaturevariance-asym}
Let $\mathcal{E}$ denote the bitstring at the node $a$ and $\prob(\cond_{reg} | \cond) > 1 - n^{-\Omega(\log n)}$. Then, $\ex[\widehat{s}_{a,i}^2|\cond] = O(\log^2 n)$ as long as $e^{2\lambda_{max}} < 2$.
\end{lemma}

\begin{proof}
Note that the sum over all $x, y$ pairs in $A$ such that $x \wedge y$ is $m$ edges away from $a$ of $2^{h(x)+h(y)}$ is $2^{h(a)-m}$. Then, the proof follows analogously to the proof of Lemma~\ref{lemma:signaturevariance}.
\end{proof}







As before, we define:

\[\widehat{C}(a,b) = \frac{2}{L} \sum_{i=1}^{L/2} \widehat{s}_{a,2i+1}\widehat{s}_{b,2i+1}\]

For some height $h = \delta\log \log n$, we define $\widehat{d}(a,b)$ similarly to before, but splitting the exact definition of $\widehat{d}(a,b)$ into three cases:

\textbf{Case 1:}\textit{ If $a, b$ both have no leaf descendants less than $h$ edges away from them}, consider the nodes $A_h$ and $B_h$ which are the set of descendants of $a, b$ exactly $h$ edges below $a, b$ respectively. Order $A_h$ arbitrarily and let $a_j$ be nodes of $A_h$ with $1 \leq j \leq 2^h = \log^{\delta} n$ and similarly for $B_h$. Again let $\widetilde{d}(a_j,b_j) = -\ln(e^{d(a_j,a)+d(b_j,b)}4\widehat{C}(a_j, b_j))$. Let $S_h(a,b) = \set{\widetilde{d}(a_j,b_j)}_{j=1}^{2^h}$, and 

\[r_j = \inf\set{r > 0 : \left|\{j' \neq j: |\widetilde{d}(a_j, b_j) - \widetilde{d}(a_{j'}, b_{j'}) | \leq r  \}\right| \geq \frac{2}{3}2^{h}}\].

Then, if $j^* = \arg\min_j r_j$, then our distance estimator is $\widehat{ d}(a,b) = \widetilde{d}(a_{j^*},b_{j^*})$.

\textbf{Case 2:} \textit{If $a$ has no leaf descendants less than $h$ edges away but $b$ has some leaf descendant $b'$ which is less than $h$ edges from $b$}, order $A_h$ arbitrarily and let $a_j$ be nodes of $A_h$ with $1 \leq j \leq 2^h = \log^{\delta} n$. Let $\widetilde{d}(a_j,b') = -\ln(e^{d(a_j,a)+d(b',b)}4\widehat{C}(a_j, b'))$, and $S_h(a,b) = \set{\widetilde{d}(a_j,b')}_{j=1}^{2^h}$, and define $r_j$ and $\hat{d}(a,b)$ analogously to Case 1.

\textbf{Case 3:} \textit{If $a$ and $b$ both have leaf descendants $a', b'$ less than $h$ edges away}, we just define $\hat{d}(a,b) = -\ln(e^{d(a',a)+d(b',b)}4\Dt'(a', b'))$.

\begin{lemma}[Deep Distance Computation: Small Diameter]
\label{lemma:smalldiameter-asym}
Let $a, b$ be nodes such that $ d(a, b) = O(\log \log n)$. If $\lambda_{max} < \ln(\sqrt{2})$, then $|\widehat{d}(a,b) - d(a,b)| < \epsilon$ with high probability.
\end{lemma}

\begin{proof}
In Case 1, the proof follows as did the proof of Lemma~\ref{lemma:smalldiameter} (including an analogous proof of Lemma~\ref{lemma:covariance}). In Case 2, the proof follows similarly to the proof of Lemma~\ref{lemma:smalldiameter}, except we also condition on the bitstring at $b'$ and note that $d(a, b') = O(\log \log n)$. In Case 3, the proof follows directly from Corollary~\ref{cor:estimatorconcentration-asym}.
\end{proof}

\subsection{Reconstruction Algorithm}

Since we have proven statements analogous to those needed to prove Theorem~\ref{thm:symmetric}, the proof of Theorem~\ref{thm:asymmetric} follows very similarly to Theorem~\ref{thm:symmetric}, except that a short quartet needs to be slightly redefined for our purposes. 

\begin{definition}
A quartet  $Q = \set{a,b,c,d}$ is {\bf r-short} corresponding to a distance matrix $D$ if for every pair $x, y \in Q$, 
\begin{itemize}
    \item When $x, y$ both have leaf descendants $x', y'$ less than $\delta\log \log n$ away, we use $D(x,y) = \widehat{d}(x,y)$ and $D(x,y) \leq r$.
    \item Otherwise, we use $D(x,y) = \widehat{d}(x,y)$ and $T(S_h(x,y), r) = 1$
\end{itemize} 
\end{definition}

By Lemma~\ref{lemma:smalldiameter-asym}, we see that $O(\log \log n)$-short quartets can be detected and \code{FPM} on these quartets always return the true quartet tree, by an analogous argument to the symmetric case. Note that in the asymmetric case, at each step of the tree reconstruction process, not all nodes will be paired as cherries but at least one cherry will be paired (by looking at the cherry with maximum depth) and we can therefore always ensure progress. 

There is, however, a slight issue with directly following the the same reconstruction algorithm because we may join subtrees that are no longer both dangling, which intuitively means that the path between the root of both subtrees goes above them in the real model tree. For example, in the case of balanced trees, if $S_1, S_2$ are subtrees that are to be joined at some iteration of the algorithm, then we reconstruct the shared ancestor of $S_1, S_2$ and join the roots of $S_1, S_2$ as children of the reconstructed ancestor. However, in this case, it might be possible that the ancestor of $S_1$ is a node in $S_2$ that is not the root node of $S_2$!

This non-dangling issue is elaborated in the general tree reconstruction algorithm of \cite{roch2008sequence} and is circumvented with standard reductions to the Blindfolded Cherry Picking algorithm of \cite{daskalakis2011evolutionary}, which essentially allows us to reduce all subtree joining processes to the dangling case. The basic idea is that there exists a re-rooting of our subtrees that reduces to the dangling case and finding the correct re-rooting boils down to some $O(1)$ extra distance computations per iteration. Because our algorithm is a close replica of the general tree reconstruction algorithm of Roch, we refer the reader to the appendix of \cite{roch2008sequence} for details.
\section{Future Directions}

In this paper, we give reconstruction guarantees which are optimal (up to the choice of constants $\alpha, \beta, \kappa$) for a popular model of the phylogenetic reconstruction problem. However, we did not attempt to optimize the constant $\kappa$ in the exponent of our sequence length requirement. In particular, we note that while we have $k$ bits of information for each sequence, we only use $\tilde{O}(\sqrt{k})$ bits of information about each sequence in our algorithm, so there is some reason to believe methods similar to ours cannot achieve the optimal value of $\kappa$. It is an interesting problem to design an algorithm which uses $\Omega(k)$ bits of information and matches our asymptotic guarantees with potentially better constants, but doing so seems challenging given the presence of indels. Furthermore, we operated in the $\Delta$-branch model with discretized edge lengths, which avoids errors accumulating over a series of distance estimations. Without discretized edge lengths, a more refined analysis seems necessary in order to avoid this error accumulation.

In addition, there are other models of theoretical or practical interest, but for which the extension from our results is not immediately obvious. One alternative model which has been studied in the trace reconstruction problem (see e.g. \cite{HoldenPP18}) and which could be extended to the phylogenetic reconstruction problem is to view the bitstrings as infinitely long. The new goal is to design an algorithm which only views the first $k(n)$ bits of each bitstring for as small a function $k(n)$ as possible. This model is well-motivated by practical scenarios, where DNA sequences are large but reading the entire sequence is both inefficient and unnecessary for reconstructing the tree. When we assume the bitstrings are infinitely long, then reconstruction may be possible without the assumptions we made to ensure no leaf bitstrings were empty (i.e., it may be possible to reconstruct trees with maximum depth $\Omega(n)$ or with large differences in the insertion and deletion rates). In Section~\ref{sec:assym} we crucially used sequence lengths to estimate the positions of blocks in the bitstrings, so even with these assumptions our results do not easily extend to this model.

Lastly, our results are optimal up to constants and build on many techniques for phylogenetic reconstruction with independent and random mutations. However, algorithms which perform well on simulated data generated using independent and random mutations are known to perform relatively poorly on real-world data \cite{NSW18}. An interesting problem is to define a theoretical model for phylogenetic reconstruction with dependent mutations or semi-adversarial mutations which better models this real-world data and to design algorithms for this model. In particular, one advantage of our algorithm is that it uses statistics about large-length blocks of bits which are very robust to the errors introduced by indels. One might hope that this robustness extends to models with dependent and/or semi-random mutations.
\section*{Acknowledgements}

We are thankful to Satish Rao for discussions on the generalization to imbalanced trees and other helpful comments, Efe Aras for helpful comments on the presentation, an anonymous reviewer for pointing out an error in the proofs of Lemmas~\ref{lemma:estimatorconcentration} and \ref{lemma:smalldiameter} in a previous version of the paper, and other anonymous reviewers for suggested revisions.

\bibliographystyle{alpha}
\bibliography{sample}

\appendix
\section{Deferred Proofs from Section~\ref{sec:sym}}\label{sec:deferred}

\begin{proof}[Proof of Lemma~\ref{lemma:lengths}]
Let $k_a$ be the length of the bitstring at node $a$. For any edge $e = (a, b)$ where $a$ is the parent, $k_b$ is equal to $k_a$ plus the difference between two binomial variables with $k_a$ trials and probability of success $\pID(e)$. Applying Azuma's inequality shows that $|k_b - k_a|$ is at most $2 \log n \sqrt{k_a}$, with probability $1 - n^{-\Omega(\log n)}$. Then fixing any node $v$ and applying this high probability statement to at most $\log n$ edges on the path from the root to any node $a$ gives that $k < k_a < 4k$ with probability $1 - n^{-\Omega(\log n)}$. The lemma then follows from a union bound over all $n$ nodes.
\end{proof}
\begin{proof}[Proof of Lemma~\ref{lemma:bitshifts}]

Note that conditioned on the $j$th bit of $a$ not being deleted on an edge $e = (a, b)$ where $a$ is the parent, the number of bits by which it shifts on the edge $(a, b)$ is the difference between two binomial variables with $j$ trials and probability of success $\pID(e)$, which is at most $2 \log n\sqrt{j}$ with probability $1 - n^{-\Omega(\log n)}$. By Lemma~\ref{lemma:lengths}, with probability $1 - n^{-\Omega(\log n)}$ we know that $j \leq 4k$ so this is at most $4\log n \sqrt{k}$. Then, fixing any path and applying this observation to the at most $2 \log n$ edges on the path, by union bound we get that the sum of shifts is at most $4 \log^2 n \sqrt{k}$ with probability $1 - n^{-\Omega(\log n)}$. Applying union bound to all $O(n^2)$ paths in the tree gives the lemma.
\end{proof}

\begin{proof}[Proof of Lemma~\ref{lemma:uniformblocks}]
For a fixed node $a$ and any consecutive sequence of length $m$, note that each bit in the bitstring is equally likely to be 0 or 1, by symmetry, and furthermore, each bit is independent since they cannot be inherited from the same bit in ancestral bitstrings. The number of zeros in the sequence, $S_0$, can be expressed as a sum of $m$ i.i.d. Bernoulli variables. Therefore, by Azuma's inequality, we have

$$\prob(S_0  - m/2 \geq t\sqrt{m}) \leq \exp(-\Omega(t^2))$$

There are $2n-1$ nodes, and by Lemma~\ref{lemma:lengths}, for each node the number of different consecutive subsequences of the node's bitstring is $\plog$. Therefore, setting $t = \log n$ and applying a union bound over all $O(n \log^{O(1)} (n))$ subsequences gives the lemma.
\end{proof}

\begin{proof}[Proof of Lemma~\ref{lemma:unbiasedestimator}]
Fixing any block $i$ of nodes $a, b$ and letting $\sigma_{a,j}$ denote the $j$th bit of the bitstring at $a$:
$$\ex[s_{a, i}s_{b, i}] 
= \frac{1}{l}\ex\left[\left(\sum_{j = (i-1)l+1}^{il} \sigma_{a,j} - \frac{l}{2}\right)\left(\sum_{j' = (i-1)l+1}^{il} \sigma_{b,j'} - \frac{l}{2}\right)\right]$$
$$= \frac{1}{l}\ex\left[\left(\sum_{j = (i-1)l+1}^{il} \left(\sigma_{a,j} - \frac{1}{2}\right)\right)\left(\sum_{j' = (i-1)l+1}^{il} \left(\sigma_{b,j'} - \frac{1}{2}\right)\right)\right]$$
$$ = \frac{1}{l}\ex\left[\sum_{j = (i-1)l+1}^{il}\sum_{j' = (i-1)l+1}^{il}\left(\sigma_{a,j} - \frac{1}{2}\right)\left(\sigma_{b,j'} - \frac{1}{2}\right)\right] $$
$$ = \frac{1}{l}\sum_{j = (i-1)l+1}^{il}\sum_{j' = (i-1)l+1}^{il}\ex\left[\left(\sigma_{a,j} - \frac{1}{2}\right)\left(\sigma_{b,j'} - \frac{1}{2}\right)\right] $$

 Note that if bit $j$ of $a$'s bitstring and bit $j'$ of $b$'s bitstring are not shared, then their values are independent and in particular, since $\ex[\sigma_{a,j}] = 1/2$ for any $a, j$:

$$\ex\left[\left(\sigma_{a,j} - \frac{1}{2}\right)\left(\sigma_{b,j'} - \frac{1}{2}\right)\right] = \ex\left[\left(\sigma_{a,j} - \frac{1}{2}\right)\right]\ex\left[\left(\sigma_{b,j'} - \frac{1}{2}\right)\right] = 0$$

Let $\cond$ be any realization of the locations where insertions and deletions occur throughout the tree. We will look at $\ex[s_{a, i}s_{b, i}|\cond]$, leaving the root bitstring, the values of bits inserted by insertions, and the locations of substitutions unrealized. Note that $\cond$ fully specifies what bits are shared by block $i$ of $a, b$, giving:

$$\ex[s_{a, i}s_{b, i}|\cond] 
= \frac{1}{l}\sum_{\textnormal{shared }j,j'}\ex\left[\left(\sigma_{a,j} - \frac{1}{2}\right)\left(\sigma_{b,j'} - \frac{1}{2}\right)\right]$$

Now, note that $\left(\sigma_{a,j} - \frac{1}{2}\right)\left(\sigma_{b,j'} - \frac{1}{2}\right)$ is $1/4$ if $\sigma_{a,j}, \sigma_{b,j'}$ are the same and $-1/4$ otherwise. Since bit $j$ of $a$ and bit $j'$ of $b$ descended from the same bit, it is straightforward to show (see e.g., \cite{warnow-textbook}) that the probability $a, b$ are the same is $\frac{1}{2}(1 + \prod_{e \in P_{a,b}} (1 - 2\pS(e)))$, giving that $\ex\left[\left(\sigma_{a,j} - \frac{1}{2}\right)\left(\sigma_{b,j'} - \frac{1}{2}\right)\right] = \frac{1}{4}\prod_{e \in P_{a,b}} (1 - 2\pS(e))$ if $j, j'$ are shared bits of $a, b$.

Applying the law of total probability to our conditioning on $\cond$, we get that $\ex[s_{a, i}s_{b, i}]$ is the expected number of shared bits in block $i$ of $a$ and block $i$ of $b$ times $\frac{1}{4l}\prod_{e \in P_{a,b}} (1 - 2\pS(e))$. So all we need to do is compute the expected number of shared bits which are in block $i$ of $a$ and $b$. Let $a \wedge b$ be the least common ancestor of $a,b$. The $j$th bit in $a \wedge b$ will not be deleted on the path from $a \wedge b$ to $a$ or $b$ with probability $\prod_{e \in P_{a,b}} (1-\pD(e))$. Let $\rho_j$ be the probability that the $j$th bit of $a \wedge b$ appears in the $i$th block of both $a$ and $b$ conditioned on it not being deleted. Then the expected number of shared bits is $(\sum_j \rho_j) \cdot \prod_{e \in P_{a,b}} (1-\pD(e)) $.

For our fixed block $i$, call the $j$th bit of $a \wedge b$ a \textit{good} bit if $j$ is between $(i-1)l+ 4\log^2 n\sqrt{k}$ and $il-4\log^2 n\sqrt{k}$ inclusive. Call the $j$th bit an \textit{okay} bit if $j$ is between $(i-1)l-4\log^2 n\sqrt{k}$ and $il+4\log^2 n\sqrt{k}$ inclusive but is not a good bit. If the $j$th bit is not good or okay, call it a \textit{bad} bit. Note that $4\log^2 n\sqrt{k} \leq l \cdot O(\log^{-\kappa\zeta+2} n)$, which is $o(l)$ if $\kappa$ is sufficiently large and $\zeta$ is chosen appropriately. Then, there are $l\cdot(1-O(\log^{-\kappa\zeta+2} n))$ good bits and $l \cdot O(\log^{-\kappa\zeta+2} n)$ okay bits for block $i$.
Lemma~\ref{lemma:bitshifts} gives that $\rho_j \geq 1-n^{-\Omega(\log n)}$  for all good bits. Similarly, $\rho_j \leq n^{-\Omega(\log n)}$ for all bad bits. For okay bits, we can lazily upper and lower bound $\rho_j$ to be in $[0, 1]$. This gives:

\begin{equation*}
\begin{aligned}
\sum_j \rho_j &= 
\sum_{\textnormal{good }j} \rho_j+\sum_{\textnormal{okay }j} \rho_j+\sum_{\textnormal{bad }j} \rho_j \\
&=l(1 - O(\log^{-\kappa\zeta+2} n))+ l \cdot O(\log^{-\kappa\zeta+2} n)+n^{-\Omega(\log n)}= l(1 \pm O(\log^{-\kappa\zeta+2} n))
\end{aligned}
\end{equation*}
 Combining this with the previous analysis gives that 
$$  \ex[s_{a, i}s_{b, i}]=\frac{1}{4}(1 \pm O(\log^{-\kappa\zeta+2} n))\prod_{e \in P_{a,b}}(1-2\pS(e))(1-\pD(e)) $$ 

Rewriting this in exponential form and using the definition of $\lambda(e)$ and $d(a,b) = \sum_{e \in P_{a,b}} \lambda(e)$ concludes our proof.
\end{proof}

\begin{proof}[Proof of Lemma~\ref{lemma:estimatorconcentration}]
We show how to bound the probability of the error in one direction, the other direction follows similarly. 

Let $j_{a,i}$ be the index where the $(i-1)l+1$th bit of the bitstring of $a \wedge b$, $a$ and $b$'s least common ancestor, ends up in the bitstring at $a$ (or if it is deleted on the path from $a \wedge b$ to $a$, where it would have ended up if not deleted). i.e., the bits in the $i$th block of $a \wedge b$ appear in positions $j_{a,i}$ to $j_{a,i+1}-1$ of the bitstring at $a$.

Let $s_{a,i}^*$ be defined analogously to $s_{a,i}$, except instead of looking at the bits in the $i$th block of $a$, we look at bits $j_{a,i}$ to $j_{a,i+1}-1$ (we still use the multiplier $\frac{1}{\sqrt{l}}$). Note that conditioned on $\cond_{reg}$, $s_{a,i}^*$ and $s_{a,i}$ differ by $O(\frac{k^\frac{1}{4} \log^2 n}{\sqrt{l}}) = O(\log^{-\kappa\zeta/2+2} n)$, so $s_{a,i}^*s_{b,i}^*$ and $s_{a,i}s_{b,i}$ differ by $O(\log^{-\kappa\zeta/2+3} n)$. Furthermore, $s_{a,i}^*s_{b,i}^*$ is completely determined by the bits in the $i$th block of $a \wedge b$ and substitutions, insertions, and deletions on the path from $a$ to $b$ in positions corresponding to the $i$th block of $a \wedge b$. For $i \neq i'$, these sets of determining random variables are completely independent, so the random variables $\{s_{a,i}^*s_{b,i}^*\}_i$ are independent. 

Define $\Dt^*(a, b)$ analogously to $\Dt(a, b)$, except using $s_{a,i}^*$ instead of $s_{a,i}$. $\Dt^*(a,b)$ and $\Dt(a,b)$ (and their expectations conditioned on $\cond_{reg}$) differ by $O(\log^{-\kappa\zeta/2+3} n)$. By rearranging terms and applying Lemma~\ref{lemma:unbiasedestimator} we get:

\begin{equation}\label{eq:sdtest1}
\begin{aligned}
\prob&[-\ln(4\Dt(a, b)) >  d(a,b) + \epsilon] \\
=\prob&[4\Dt(a, b) < e^{- d(a,b) - \epsilon}] \\
=\prob&[\Dt(a, b) < \frac{1}{4}e^{- d(a,b)} - \frac{1}{4}(1 - e^{- \epsilon})e^{-  d(a,b)}] \\
=\prob&[\Dt(a, b) < \frac{1}{4}(1 \pm O(\log^{-\kappa\zeta+2} n))e^{-  d(a,b)} - \frac{1}{4}(1 - e^{- \epsilon} \pm O(\log^{-\kappa\zeta+2} n))e^{-  d(a,b)}] \\
=\prob&[\Dt(a, b) < \ex[\Dt(a,b)] - \frac{1}{4}(1 - e^{- \epsilon} \pm O(\log^{-\kappa\zeta+2} n))e^{-  d(a,b)}] \\
=\prob&[\Dt(a, b) < \ex[\Dt(a, b) | \cond_{reg}] - (1 - e^{- \epsilon} \pm O(\log^{-\kappa\zeta+2} n))(\frac{1}{4}e^{-  d(a,b)}) ] \\
\leq \prob&[\Dt(a, b) < \ex[\Dt(a, b) | \cond_{reg}] - (1 - e^{- \epsilon} \pm O(\log^{-\kappa\zeta+2} n))(\frac{1}{4}e^{-  d(a,b)}) |\cond_{reg}] + n^{-\Omega(\log n)}\\
\leq \prob&[\Dt^*(a, b) < \ex[\Dt^*(a, b) | \cond_{reg}] - (1 - e^{- \epsilon} \pm O(\log^{-\kappa\zeta/2+\delta+3} n))(\frac{1}{4}e^{-  d(a,b)}) |\cond_{reg}] + n^{-\Omega(\log n)}
\end{aligned}
\end{equation}

Note that conditioned on $\cond_{reg}$ no $s^*_{a, i}s^*_{b, i}$ exceeds $O(\log^2 n)$ in absolute value, so the difference in $\ex[\Dt^*(a, b)]$ induced by conditioning on an additional value of $s_{a, 2i+1}s_{b, 2i+1}$ is $O(\log^2 n/L)$. Azuma's and an appropriate choice of $\kappa$ then gives: 

\begin{equation}\label{eq:sdtest2}
\begin{aligned} 
& \prob[\Dt^*(a, b) < \ex[\Dt^*(a, b) | \cond_{reg}] - (1 - e^{- \epsilon} \pm O(\log^{-\kappa\zeta/2+\delta+3} n))\frac{1}{4}e^{-  d(a,b)} |\cond_{reg}] \\
\leq &\exp\left(-\frac{((1 - e^{- \epsilon} \pm O(\log^{-\kappa\zeta/2+\delta+3} n))\frac{1}{4}e^{-  d(a,b)} )^2}{(L/ 2 - 1)O(\log^2 n/L)^2}\right) \\
= & \exp(-\Omega(L  \epsilon e^{-2  d(a,b)}/\log^4 n))\\
\leq & \exp(-\Omega(\epsilon \log^{\kappa(1/2-\zeta)-2 \delta - 4}n)) \\
\leq& n^{-\Omega(\log n)}
\end{aligned}
\end{equation}

Combining \eqref{eq:sdtest1} and \eqref{eq:sdtest2} gives the desired bound.

\end{proof}

\begin{proof}[Proof of Lemma~\ref{lemma:recursiveestimator}]

We will implicitly condition on $\cond_{reg}$ in all expectations in the proof of the lemma. 

Let $S_1$ be the set of bits in block $i$ of $a$ which appear in $x$'s bitstring. Let $S_2$ be the set of bits which appeared anywhere in $a$'s bitstring except block $i$ of $a$'s bitstring and which are in block $i$ of $x$'s bitstring. Let $S_3$ be the set of bits from block $i$ of $a$ which appear in $x$'s bitstring outside of block $i$ (note that $S_3$ is a subset of $S_1$).

For $j \in \{1, 2, 3\}$, consider the values of the bits in $S_j$ in $x$'s bitstring. Let $s_{x, i}^{(j)}$ denote the number of zeroes in the bits in $S_j$ minus $|S_j|/2$, all times $1/\sqrt{l}$. Note that because all bits which are present in $x$ but not in $a$ are uniformly random conditioned on $\mathcal{E}$, $\ex[s_{x,i}] = \ex[s_{x,i}^{(1)} + s_{x,i}^{(2)} - s_{x,i}^{(3)}]$. Informally, this equality says that the bits determining $s_{x,i}$ are going to be those in block $i$ of $a$ that survive the deletion process, except those that are moved out of block $i$ by the indel process, and also including bits moved into block $i$ by the indel process.

By a similar argument as in Lemma~\ref{lemma:unbiasedestimator}, $\ex[s_{x,i}^{(1)}]$ is exactly $s_{a,i}e^{- d(a,x)}$, by taking into account the probability a bit survives $h$ levels of the indel process times the decay in the expectation of every bit's contribution induced by the the substitution process. 

Now, consider the bits in $S_2$. For bit $j$ of $a$'s bitstring, such that bit $j$ is not in the $i$th block, let $\sigma_{a,j}$ be the value of this bit in $a$'s bitstring as before. Let $\rho_j$ denote the probability that this bit appears in block $i$ of $x$, conditioned on $j$ not being deleted between $a$ and $x$. The expected contribution of the bit $j$ to $s_{x,i}^{(2)}$ is then $(\sigma_{a,j}-1/2)\rho_j e^{-d(a,x)} / \sqrt{l}$ (the $e^{-d(a,x)}$ is again due to decay in expected contribution induced by substitution and deletion).

Now, by linearity of expectation: 

$$\ex[s_{x,i}^{(2)}] = \frac{e^{-d(a,x)}}{\sqrt{l}}\left[\sum_{j < (i-1)l+1} (\sigma_{a,j}-1/2)\rho_j + \sum_{j > il} (\sigma_{a,j}-1/2)\rho_j\right]$$

We will restrict our attention to showing that the sum $\sum_{j > il} (\sigma_{a,j}-1/2)\rho_j$ is sufficiently small, the analysis of the other sum is symmetric. Since $\prob(\cond_{reg}|\cond) > 1- O(n^{-\Omega(\log n)})$, we know that for $j > il + 4\log^{2}(n)\sqrt{k} =: j^*$, $\rho_j = n^{-\Omega(\log n)}$. So:

$$\sum_{j > il} (\sigma_{a,j}-1/2)\rho_j = \sum_{j = il+1}^{j^*} (\sigma_{a,j}-1/2)\rho_j + n^{-\Omega(\log n)}$$

Then define $d_j := \rho_j - \rho_{j+1}$, $\sigma^*_{j'} = \sum_{j = il+1}^{j'} (\sigma_{a,j} - 1/2)$, and note that by the regularity assumptions $\cond_{reg}$, $\sigma_{j}^*$ is at most $\log n \sqrt{j - il + 1} \leq \log n  \sqrt{j^* - il + 1} \leq 2 \log^2 n k^{1/4}$ for all $j$ in the above sum. 
Also note that for all $j$ in the above sum, $\rho_j$ is decreasing in $j$, so the $d_j$ are all positive and their sum is at most 1. Then: 

$$\sum_{j > il} (\sigma_{a,j}-1/2)\rho_j = \sum_{j = il+1}^{j^*} d_j \sigma^*_{j'} + n^{-\Omega(\log n)}$$
$$\implies |\sum_{j > il} (\sigma_{a,j}-1/2)\rho_j| \leq \max_{il+1 \leq j \leq j^*} |\sigma^*_{j'}| + n^{-\Omega(\log n)} \leq 2 \log^2 n k^{1/4}$$

Thus $|\ex[s_{x,i}^{(2)}]| \leq 4 \log^2 n k^{1/4}/\sqrt{l}$. A similar argument shows $|\ex[s_{x,i}^{(3)}]| \leq 4 \log^2 n k^{1/4}/\sqrt{l}$, completing the proof.

\end{proof}

\begin{proof}[Proof of Lemma~\ref{lemma:covariance}]

$$\Cov(\widehat{s}_{a_j,2i+1}\widehat{s}_{b_j,2i+1}, \widehat{s}_{a_j,2i'+1}\widehat{s}_{b_j,2i'+1})$$
\begin{equation}\label{eq:cov1}
\begin{aligned} =&\ex[\widehat{s}_{a_j,2i+1}\widehat{s}_{b_j,2i+1}\widehat{s}_{a_j,2i'+1}\widehat{s}_{b_j,2i'+1}] - \ex[\widehat{s}_{a_j,2i+1}\widehat{s}_{b_j,2i+1}]\ex[\widehat{s}_{a_j,2i'+1}\widehat{s}_{b_j,2i'+1}] \\
= & \ex[\widehat{s}_{a_j,2i+1}\widehat{s}_{b_j,2i+1}\widehat{s}_{a_j,2i'+1}\widehat{s}_{b_j,2i'+1}] -\\ &(s_{a_j,2i+1}s_{b_j,2i+1} \pm O(\log^{-\kappa\zeta/2+2} n))(s_{a_j,2i'+1}s_{b_j,2i'+1} \pm O(\log^{-\kappa\zeta/2+2} n)) \\
\leq & \ex[\widehat{s}_{a_j,2i+1}\widehat{s}_{a_j,2i'+1}]\ex[\widehat{s}_{b_j,2i+1}\widehat{s}_{b_j,2i'+1}] -\\ &s_{a_j,2i+1}s_{b_j,2i+1}s_{a_j,2i'+1}s_{b_j,2i'+1} + O(\log^{-\kappa\zeta/2+4} n) \\
\end{aligned}
\end{equation}

To show the covariance is small, it thus suffices to show that $\ex[\widehat{s}_{a_j,2i+1}\widehat{s}_{a_j,2i'+1}]$ is close to $s_{a_j,2i+1}s_{a_j,2i'+1}$. Using the definition of $\widehat{s}$ we get:

\begin{equation}\label{eq:cov2}
\ex[\widehat{s}_{a_j,2i+1}\widehat{s}_{a_j,2i'+1}] = \frac{1}{|A|^2} \sum_{x, x' \in A_j} e^{d(x, a_j) + d(x', a_j)}\ex[s_{x, 2i+1}s_{x',2i'+1}]
\end{equation}

Isolating one of the terms $\ex[s_{x, 2i+1}s_{x',2i'+1}]$:

\begin{equation}\label{eq:cov3}
\ex[s_{x, 2i+1}s_{x',2i'+1}] = \ex[\frac{1}{l} \sum_{j, j'} (\sigma_{x,m} - 1/2)(\sigma_{x',m'} - 1/2)]
\end{equation}

Where the sum is over indices $m$ in block $2i+1$ and indices $m'$ in block $2i'+1$. Note that

$$s_{a_j, 2i+1}s_{a_j, 2i'+1} = \frac{1}{l} \sum_{m, m'} (\sigma_{a_j,m} - 1/2)(\sigma_{a_j,m'} - 1/2)$$

Let $\rho_{m,m'}$ denote the probability that bit $m$ of $a_j$ appears in block $i$ of $x$ and bit $m'$ of $a_j$ appears in block $i'$ of $x'$ (conditioned on each not being deleted). For a fixed block, again define a bit to be \textit{good} if it is within the block and more than $4\log^2 n\sqrt{k}$ bits away from the boundaries of the block, \textit{okay} if it is within $4\log^2 n\sqrt{k}$ bits of the boundaries of the block, and \textit{bad} otherwise. Using a similar argument to the proof of Lemma~\ref{lemma:unbiasedestimator} we get:

$$\ex[s_{x, 2i+1}s_{x',2i'+1}] =\frac{1}{l}e^{-d(x, a_j)-d(x', a_j)} \sum_{m, m'} (\sigma_{a_j,m} - 1/2)(\sigma_{a_j,m'} - 1/2)\rho_{m,m'}+n^{-\Omega(\log n)}$$

Where the sum is over all indices $m, m'$ where $j$ is not bad for block $i$ and $j'$ is not bad for block $i'$. If either index is bad, then $\rho_{m,m'} \leq n^{-\Omega(\log n)}$ and the pairs' contribution to the sum is  absorbed into the second term. Note that $e^{-d(x, a_j)-d(x', a_j)}$ is not the correct decay term when $m = m'$, but in this case the index is bad for one of the blocks and can be absorbed into the $n^{-\Omega(\log n)}$ term. 

By Lemma~\ref{lemma:bitshifts} and a union bound if $m$ is good for block $i$ and $m'$ is good for block $i'$, $\rho_{m,m'} \geq 1 - n^{-\Omega(\log n)}$. Otherwise, we naively bound $\rho_{m,m'}$ between 0 and 1. Using Lemma~\ref{lemma:uniformblocks}, we know that $|\sum_{\text{okay }m}(\sigma_{x \wedge x', m}-1/2)| \leq 4 \log^2 n \cdot k^{1/4}, |\sum_{\text{good }m}(\sigma_{x \wedge x', m}-1/2)| \leq \sqrt{l} \log n$ with high probability. So the contribution of terms for these $m, m'$ pairs to $\ex[s_{x, 2i+1}s_{x',2i'+1}]$ is $O(\frac{\log^3 n k^{1/4}}{\sqrt{l}}) \cdot e^{-d(x,a_j) - d(x',a_j)}$ in magnitude. Similarly, the contribution of such terms to $s_{a_j, 2i+1}s_{a_j, 2i'+1}$ is $O(\frac{\log^3 n k^{1/4}}{\sqrt{l}})$ in magnitude. This gives:

$$\ex[s_{x, 2i+1}s_{x',2i'+1}] =e^{-d(x, a_j)-d(x', a_j)}\left[\frac{1}{l} \sum_{\text{good }j, j'} (\sigma_{x \wedge x',j} - 1/2)(\sigma_{x \wedge x',j'} - 1/2)\pm O(\log^{-\kappa\zeta/2+3} n)\right]$$

$$s_{a_j, 2i+1}s_{a_j, 2i'+1} = \frac{1}{l} \sum_{\text{good }j, j'} (\sigma_{a_j,j} - 1/2)(\sigma_{a_j,j'} - 1/2)\pm O(\log^{-\kappa\zeta/2+3} n)$$

\begin{equation}\label{eq:cov4}
\implies \ex[s_{x, 2i+1}s_{x',2i'+1}] = e^{-d(x, a_j)-d(x', a_j)}[s_{a_j, 2i+1}s_{a_j, 2i'+1} \pm O(\log^{-\kappa\zeta/2+3} n)] 
\end{equation}

Combining \eqref{eq:cov1}-\eqref{eq:cov4} and using the regularity conditions gives the Lemma statement.
\end{proof}

\section{Deferred Proofs from Section~\ref{sec:assym}}\label{sec:deferred-assym}

\begin{proof}[Proof of Lemma~\ref{lemma:lengths-asym}]
Let $k_a$ be the length of the bitstring at vertex $a$. For any edge $e = (a, b)$ where $a$ is the parent, $k_b$ is equal to $k_a$ plus the difference between the number of insertions and deletions on this edge. Applying Azuma's inequality separately to the number of insertions and deletions gives that $|k_b/\eta(b) - k_a/\eta(a)|$ is at most $2 \log n \sqrt{k_a}/\eta(a) \leq 2 \log^{\beta/2+1} n \sqrt{k_a/\eta(a)} $ with probability $1 - n^{-\Omega(\log n)}$. This can be rewritten as:

$$\frac{k_a}{\eta(a)}(1 - 2 \log^{\beta/2+1} n \sqrt{\frac{\eta(a)}{k_a}}) \leq k_b/\eta(b) \leq \frac{k_a}{\eta(a)}(1 + 2 \log^{\beta/2+1} n \sqrt{\frac{\eta(a)}{k_a}})$$

Let $\cond$ denote the event that this happens for every edge. By a union bound, this happens with probability $1 - n^{-\Omega(\log n)}$.

The lemma statement of course holds for the root node. For $b$'s parent $a$, inductively assume that $k_r (1 - \frac{\depth(a)\cdot \log^{\beta/2+1}(n)}{\log^{\kappa/2-2} n}) \leq k_a/\eta(a) \leq k_r (1 + \frac{\depth(a)\cdot \log^{\beta/2+1}(n)}{\log^{\kappa/2-2} n})$. Then conditioned on $\cond$:

$$k_b/\eta(b) \leq \frac{k_a}{\eta(a)}(1 + 2 \log^{\beta/2+1} n \sqrt{\frac{\eta(a)}{k_a}})$$
$$k_b/\eta(b) \leq k_r (1 + \frac{\depth(u)\cdot \log^{\beta/2+1}(n)}{\log^{\kappa/2-2} n})(1 + 2 \log^{\beta/2+1} n \sqrt{\frac{1}{k_r (1 - \frac{\depth(a)\cdot \log^{\beta/2+1}(n)}{\log^{\kappa/2-2} n})}})$$
$$k_b/\eta(b) \leq k_r (1 + \frac{\depth(a)\cdot \log^{\beta/2+1}(n)}{\log^{\kappa/2-2} n})(1 + 2 \log n \sqrt{\frac{1}{\Theta(\log^{\kappa} n)}})$$
$$k_b/\eta(b) \leq k_r (1 + \frac{(\depth(a)+1)\cdot \log^{\beta/2+1}(n)+1}{\log^{\kappa/2-2} n}) = k_r (1 + \frac{\depth(b) \cdot \log^{\beta/2+1}(n)+1}{\log^{\kappa/2-2} n})$$

We use the fact that for fixed $\alpha, \beta$, for sufficiently large $\kappa$, $(1 - \frac{\depth(a)\cdot \log^{\beta/2+1}(n)}{\log^{\kappa/2-2} n}) = \Theta(1)$. The bound in the other direction follows similarly, so by induction the lemma holds.

\end{proof}

\begin{proof}[Proof of Lemma~\ref{lemma:bitshifts-asym}]

Note that conditioned on the $j$th bit of $a$ not being deleted on an edge $(a, b)$ where $a$ is the parent, its normalized shift on the edge $e = (a, b)$ is the random variable:

$$\left|\frac{j+\textnormal{Binom}(j, \pI(e))-\textnormal{Binom}(j, \pD(e))}{\eta(b)} - \frac{j}{\eta(a)}\right|$$
$$ = \frac{1}{\eta(b)} \left|\textnormal{Binom}(j, \pI(e))-\textnormal{Binom}(j, \pD(e))- j(\pI(e)-\pD(e))\right|$$

By applying Azuma's inequality to $\textnormal{Binom}(j, \pI(e))-\textnormal{Binom}(j, \pD(e))$, we get that it differs from  $j(\pI(e)-\pD(e))$ by at most $\log n\sqrt{j}$ with probability $1 - n^{-\Omega(\log n)}$. By Lemma~\ref{lemma:lengths-asym}, with probability $1 - n^{-\Omega(\log n)}$ we know that $j \leq 4k_r\eta(b)$ so the normalized shift is at most $2 \log n \sqrt{k_r}$. Then, fixing any path and applying this observation to the at most $\depth \leq 2 \log^\alpha n$ edges on the path, by union bound we get that the sum of normalized shifts is at most $4 \log ^{\alpha+1} n \sqrt{k_r}$ with probability $1 - n^{-\Omega(\log n)}$. Applying union bound to all $O(n^2)$ paths in the tree gives the lemma.
\end{proof}

\begin{proof}[Proof of Lemma~\ref{lemma:signatureestimation}]
Note that $l_a' = k_a/L$ and thus by Lemma \ref{lemma:lengths-asym}:
$$k_r  (1 - \frac{\depth(a) \cdot \log^{\beta/2+1}(n)}{\log^{\kappa/2-2} n}) \leq k_a/\eta(a) \leq k_r (1 + \frac{\depth(a)\cdot \log^{\beta/2+1}(n)}{\log^{\kappa/2-2} n}) \implies$$
$$l_a  (1 - \frac{\depth(a)\cdot \log^{\beta/2+1}(n)}{\log^{\kappa/2-2} n}) \leq l_a' \leq l_a (1 + \frac{\depth(a)\cdot \log^{\beta/2+1}(n)}{\log^{\kappa/2-2} n})$$

The bits which are in exactly one of the $i$th block of $a$ and the $i$th pseudo-block of $a$ are split into two segments: bits $\min\{(i-1)l_a+1, (i-1)l_a'+1\}$ to $\max\{(i-1)l_a+1, (i-1)l_a'+1\}$ and $\min\{il_a, il_a' \}$ to $\max\{il_a, il_a'\}$. The number of bits in each segment is at most:

$$i|l_a - l_a'| \leq L \cdot  l_a \cdot \frac{\depth(a)\cdot \log^{\beta/2+1}(n)}{\log^{\kappa/2-2} n} \leq \log^{\kappa/2 - \kappa\zeta} n \cdot l_a \cdot \frac{\log^{\alpha+\beta/2+1}(n)}{\log^{\kappa/2 - 2} n} = l_a \cdot \log^{\alpha + \beta/2 + 3 - \kappa\zeta} n$$

The signed difference between the number of zeroes and expected number of zeroes in the $i$th pseudobock is $\sqrt{l_a'} \cdot \st_{a,i}$. $\sqrt{l_a} \cdot s_{a,i}$ is the same for the regular block. So applying Lemma \ref{lemma:uniformblocks} to each segment of bits:

$$|\sqrt{l_a'} \cdot \st_{a,i} - \sqrt{l_a} \cdot s_{a,i}| \leq 2 \log^{\alpha/2+\beta/4 + 5/2 - \kappa\zeta/2} n \cdot  \sqrt{l_a}$$

Combining this statement with an application of Lemma~\ref{lemma:uniformblocks} to the $i$th pseudo-block gives:

$$|\st_{a,i} - s_{a,i}| \leq |\st_{a,i} - \st_{a,i} \sqrt{l_a'/l_a}| + |\st_{a,i} \sqrt{l_a'/l_a} - s_{a,i}| \leq$$
$$\log n \sqrt{\frac{\depth(a)\cdot \log^{\beta/2+1}(n)}{\log^{\kappa/2-2} n}} + 2 \log^{\beta/4 + 3 - \kappa\zeta/2} n = O(\log^{\alpha/2+\beta/4 + 5/2 - \kappa\zeta/2} n)$$
\end{proof}

\begin{proof}[Proof of Lemma~\ref{lemma:unbiasedestimator-asym}]
Analogously to the symmetric case, letting $\cond$ denote any realization of the set of insertion and deletion locations throughout the trees (but not the root bitstring, inserted, bits or substitutions), we get:

$$\ex[s_{a, i}s_{b, i} | \cond] 
= \frac{1}{\sqrt{l_a}\sqrt{l_b}}\sum_{j,j'\textnormal{ shared by }a,b}\frac{1}{4}\prod_{e \in P_{a,b}} (1 - 2\pS(e))$$

So as before, we compute the expected number of shared bits which are in block $i$ of $a$ and $b$ and then apply the law of total probability. The probability of the $j$th bit not being deleted on the path from $a \wedge b$ to $a$ or the path from $a \wedge b$ to $b$ is again $\prod_{e \in P_{a,b}} (1 - \pD(e))$. Let $\rho_j$ is the probability the $j$th bit of $a \wedge b$'s bitstring appears in the $i$th block of $a$ and $b$'s bitstring, conditioned on the $j$th bit not being deleted on the path from $a \wedge b$ to $a$ or $b$. Then the expected number of shared bits is $(\sum_j \rho_j) \cdot \prod_{e \in P_{a,b}} (1-\pD(e)) $.

We again classify bits in the bitstring at $a \wedge b$ as good, okay, or bad. For block $i$, call the $j$th bit of $a \wedge b$ a \textit{good} bit if $j$ is between $(i-1)l_{a \wedge b}+ 4\log^{\alpha+1} (n) \sqrt{k}\eta(a \wedge b)$ and $il_{a \wedge b}- 4\log^{\alpha+1} (n) \sqrt{k}\eta(a \wedge b)$ inclusive. Call the $j$th bit an \textit{okay} bit if $j$ is between $(i-1)l_{a \wedge b}- 4\log^{\alpha+1} (n) \sqrt{k}\eta(a \wedge b)$ and $il_{a \wedge b}+ 4\log^{\alpha+1} (n) \sqrt{k}\eta(a \wedge b)$ inclusive but is not a good bit. If the $j$th bit is not good or okay, call it a \textit{bad} bit. Note that $4\log^{\alpha+1} n \sqrt{k}\eta(a \wedge b) \leq l_{a \wedge b} \cdot O(\log^{-\kappa\zeta+\alpha+1} n)$, which is $o(l_{a \wedge b})$ if $\kappa$ is sufficiently large and $\zeta$ is chosen appropriately. Then, there are $l_{a \wedge b}\cdot(1-O(\log^{-\kappa\zeta+\alpha+1} n))$ good bits and $l_{a \wedge b} \cdot O(\log^{-\kappa\zeta+\alpha+1} n)$ okay bits for block $i$.

Lemma~\ref{lemma:bitshifts-asym} gives that $\rho_j \geq 1-n^{-\Omega(\log n)}$ for good bits $j$. Similarly, $\rho_j \leq n^{-\Omega(\log n)}$ for bad bits $j$. For okay bits, we can lazily upper and lower bound $\rho_j$ to be in $[0, 1]$. Similarly to the symmetric case, we get that $\sum_j \rho_j = l_{a \wedge b} \cdot (1 \pm O(\log^{-\kappa\zeta+\alpha+1} n))$.

Combined with the previous analysis this gives:

$$  \ex[s_{a, i}s_{b, i}]=\frac{1}{4}\frac{l_{a \wedge b}}{\sqrt{l_a}\sqrt{l_b}}(1 \pm O(\log^{-\kappa\zeta+\alpha+1} n))\prod_{e \in P_{a,b}}(1-2\pS(e))(1-\pD(e))$$ 

Note that since the multiset union of $P_{r, a}$ and $P_{r, b}$ contains every edge in $P_{r, a \wedge b}$ twice and every edge in $P_{a, b}$ once:

$$l_{a \wedge b} =\sqrt{l_{a} l_{b} \cdot \prod_{e \in P_{a,b}} (1 + \pI(e) - \pD(e))^{-1}}$$ 

So we get:

$$  \ex[s_{a, i}s_{b, i}]=\frac{1}{4}(1 \pm O(\log^{-\kappa\zeta+\alpha+1} n))\prod_{e \in P_{a,b}}(1-2\pS(e))(1-\pD(e))(1 + \pI(e) - \pD(e))^{-1/2} $$ 

Rewriting this in exponential form and using the definition of $\lambda(e)$ and $d(a,b) = \sum_{e \in P_{a,b}} \lambda(e)$ concludes our proof.
\end{proof}

\begin{proof}[Proof of Lemma~\ref{lemma:recursiveestimator-asym}]

We will implicitly condition on $\mathcal{E}$ in all expectations in the proof of the lemma. 

Analogously to the proof of Lemma~\ref{lemma:recursiveestimator}, let $S_1$ be the set of bits in block $i$ of $a$ which appear in $x$'s bitstring, let $S_2$ be the set of bits which appeared anywhere in $a$'s bitstring except block $i$ of $a$'s bitstring and which are in the $i$th pseudo-block of $x$'s bitstring, and let $S_3$ be the set of bits from block $i$ of $a$ which appear in $x$'s bitstring outside of $i$th pseudo-block (note that $S_3$ is a subset of $S_1$).

For $j \in \{1, 2, 3\}$, consider the values of the bits in $S_j$ in $x$'s bitstring. Let $s_{x, i}^{(j)}$ denote the number of zeroes in the bits in $S_j$ minus $|S_j|/2$, all times $1/\sqrt{l_x}$. Note that because all bits which are present in $x$ but not in $a$ are uniformly random conditioned on $\mathcal{E}$, $\ex[\st_{x,i}] = \frac{\sqrt{l_x}}{\sqrt{l_x'}} \ex[s_{x,i}^{(1)} + s_{x,i}^{(2)} - s_{x,i}^{(3)}]$. Informally, this equality says that the bits determining $s_{x,i}$ are going to be those in block $i$ of $a$ that survive the deletion process, except those that do not appear in pseudo-block $i$ because of the indel process, and also including bits moved into pseudo-block $i$ by the indel process. 

By Lemma ~\ref{lemma:lengths-asym}, $\frac{\sqrt{l_x}}{\sqrt{l_x'}} = (1 \pm O(\log^{\alpha/2+\beta/4-\kappa/4-1/2} n))$. So it suffices to prove $\ex[s_{x,i}^{(1)} + s_{x,i}^{(2)} - s_{x,i}^{(3)}] = e^{-d(x,a)}(s_{a,i} + \nu_{a,i}')$ where $|\nu_{a,i}'| = O(\log^{5/2-\kappa\zeta/2} (n))$, since conditioned on $\cond_{reg}$ $s_{a,i}$ is at most $\log n$ in absolute value, so  $ \frac{\sqrt{l_x}}{\sqrt{l_x'}}(s_{a,i} + \nu_{a,i}') = (s_{a,i} + \nu_{a,i})$ where $|\nu_{a,i}| = O(\log^{5/2-\kappa\zeta/2} (n))$.

By a similar argument to Lemma~\ref{lemma:unbiasedestimator-asym}, $\sqrt{l_x}\ex[s_{x,i}^{(1)}] = \sqrt{l_a} s_{a,i}\prod_{e \in P_{a,x}}(1-2\pS(e))(1-\pD(e))$ and thus $\ex[s_{x,i}^{(1)}] = s_{a,i}e^{-d(x,a)}$.

Now, consider the bits in $S_2$. For bit $j$ of $a$'s bitstring, such that bit $j$ is not in the $i$th block, let $\sigma_{a,j}$ be the value of this bit in $a$'s bitstring as before. Let $\rho_j$ denote the probability that this bit appears in pseudo-block $i$ of $x$, conditioned on $j$ not being deleted between $a$ and $x$. The expected contribution of the bit $j$ to $\sqrt{l_x}s_{x,i}^{(2)}$ is then $(\sigma_{a,j}-1/2)\rho_j \prod_{e \in P_{a,x}}(1-2\pS(e))(1-\pD(e))$.

Now, by linearity of expectation: 

$$\ex[s_{x,i}^{(2)}] = \frac{\prod_{e \in P_{a,x}}(1-2\pS(e))(1-\pD(e))}{\sqrt{l_x}}\left[\sum_{j < (i-1)l_a+1} (\sigma_{a,j}-1/2)\rho_j + \sum_{j > il_a} (\sigma_{a,j}-1/2)\rho_j\right]$$
$$= e^{-d(x,a)} \cdot \frac{1}{\sqrt{l_a}} \left[\sum_{j < (i-1)l_a+1} (\sigma_{a,j}-1/2)\rho_j + \sum_{j > il_a} (\sigma_{a,j}-1/2)\rho_j\right]$$

We will restrict our attention to showing that the sum $\sum_{j > il_a} (\sigma_{a,j}-1/2)\rho_j$ is sufficiently small, the analysis of the other sum is symmetric. Since $\prob(\cond_{reg}|\cond) > 1- O(n^{-\Omega(\log n)})$, and using Lemma~\ref{lemma:lengths-asym}, we know that for $j > il_a + 4\log^{2}(n)\sqrt{k_r}\eta(a) + L \cdot l_a \cdot \frac{\log^{\alpha+\beta/2+1} n}{\log^{\kappa/2 - 2} n} =: j^*$, $\rho_j = n^{-\Omega(\log n)}$ (recall from the proof of Lemma~\ref{lemma:signatureestimation} that the term $L \cdot l_x \cdot \frac{\log^{\alpha+\beta/2+1} n}{\log^{\kappa/2 - 2} n}$ is an upper bound on the offset between the $i$th block and pseudo-block of $x$ conditioned on $\cond_{reg}$. $L \cdot l_a \cdot \frac{\log^{\alpha+\beta/2+1} n}{\log^{\kappa/2 - 2} n}$ is this offset, rescaled by the ratio of block lengths $l_a/l_x$). So:

$$\sum_{j > il_a} (\sigma_{a,j}-1/2)\rho_j = \sum_{j = il_a+1}^{j^*} (\sigma_{a,j}-1/2)\rho_j + n^{-\Omega(\log n)}$$

Then define $d_j := \rho_j - \rho_{j+1}$, $\sigma^*_{j'} = \sum_{j = il_a+1}^{j'} (\sigma_{a,j} - 1/2)$, and note that by the regularity assumptions $\cond_{reg}$, $\sigma_{j}^*$ is at most $\log n \sqrt{j - il_a + 1} \leq \log n  \sqrt{j^* - il_a + 1} \leq O(\log^{\alpha/2+\beta/4+5/2} (n)) k_r^{1/4} \sqrt{\eta(a)}$ for all $j$ in the above sum. 
Also note that for all $j$ in the above sum, we can assume $\rho_j$ is decreasing in $j$ - otherwise, we can reorder the values of $j$ in the sum so that $\rho_j$ is decreasing in this order, and redefine the $d_j$ values accordingly. By this assumption, the $d_j$ are all positive and their sum is at most 1. Then: 

$$\sum_{j > il_a} (\sigma_{a,j}-1/2)\rho_j = \sum_{j = il_a+1}^{j^*} d_j \sigma^*_{j'} + n^{-\Omega(\log n)}$$
$$\implies |\sum_{j > il_a} (\sigma_{a,j}-1/2)\rho_j| \leq \max_{il_a+1 \leq j \leq j^*} |\sigma^*_{j}| + n^{-\Omega(\log n)} \leq O(\log^{\alpha/2+\beta/4+5/2} (n) k_r^{1/4} \sqrt{\eta(a)})$$

Thus $|\ex[s_{x,i}^{(2)}]| \leq e^{-d(a,x)} O(\log^{\alpha/2+\beta/4+5/2} (n)) k_r^{1/4} \cdot \sqrt{\eta(a)}/\sqrt{l_a}$. Note that $\eta(a)/l_a = 1/l_r = O(\log^{-\kappa/2-\kappa\zeta} (n))$, giving that $|\ex[s_{x,i}^{(2)}]| = e^{-d(a,x)} \cdot O(\log^{\alpha/2+\beta/4+5/2-\kappa\zeta/2} (n))$. A similar argument shows $|\ex[s_{x,i}^{(3)}]| = e^{-d(a,x)} \cdot  O(\log^{\alpha/2+\beta/4+5/2-\kappa\zeta/2} (n))$, completing the proof.

\end{proof}

\end{document}